\documentclass[11pt]{article}
\usepackage[utf8]{inputenc}
\usepackage{amsfonts,epsfig}
\usepackage[hyphens]{url}
\RequirePackage{color}
\usepackage{mathtools}
\usepackage{extarrows}
\usepackage{enumerate} 
\usepackage{hyperref}
\hypersetup{
	colorlinks=true,
	linkcolor=black,
	citecolor=black,
	filecolor=black,
	urlcolor=black,
}
\usepackage[toc,page]{appendix}

\usepackage{breakurl}
\usepackage{comment}
\graphicspath{{fig/}}
\usepackage{wrapfig}
\usepackage{amsthm}

\newtheorem{proposition}{Proposition}
\usepackage[ruled,vlined]{algorithm2e}
\usepackage{relsize}
\let\oldv\verbatim
\let\oldendv\endverbatim
\def\verbatim{\par\setbox0\vbox\bgroup\oldv}
\def\endverbatim{\oldendv\egroup\fboxsep0pt \noindent\colorbox[gray]{0.96}{\usebox0}\par}
\usepackage{amssymb}
 \usepackage{geometry}
\usepackage{sectsty}
\usepackage{caption}
\usepackage{subcaption}
\usepackage{amsmath}

\DeclareMathOperator{\E}{\mathbb{E}}

\DeclareMathOperator{\R}{\mathbb{R}}

\sectionfont{\sffamily\upshape\large}
\subsectionfont{\sffamily\upshape\normalsize}
\subsubsectionfont{\sffamily\mdseries\upshape\normalsize}
\makeatletter
\renewcommand\@seccntformat[1]{\csname the#1\endcsname.\quad}
\makeatother
\usepackage[round]{natbib}

\makeatletter
\def\@maketitle{%
	\begin{center}%
		\let \footnote \thanks
		{\large \@title \par}%
		{\normalsize
			\begin{tabular}[t]{c}%
				\@author
			\end{tabular}\par}%
		{\small \@date}%
	\end{center}%
}
\makeatother

\title{
	\bf  Adaptive Path Sampling in Metastable Posterior Distributions 
	\vspace{.1in} 
}
\author{Yuling Yao \footnote{joint first author.} \footnote{Department of Statistics, Columbia University.}
   \and  Collin  Cademartori \footnotemark[1] \footnotemark[2]
	\and Aki Vehtari \footnote{Helsinki Institute of Information Technology; Aalto University, Department of Computer Science.}
	\and Andrew Gelman \footnote{Department of Statistics and Political Science, Columbia University.}
}
\date{\vspace{.1in} 1 Sep 2020\vspace{-.15in} }
\begin{document}\sloppy
	\maketitle
	 \thispagestyle{empty}

\begin{abstract} 
The normalizing constant plays an important role in Bayesian computation, and there is a large literature on methods for computing or approximating normalizing constants that cannot be evaluated in closed form.  When the normalizing constant varies by orders of magnitude, methods based on importance sampling can require many rounds of tuning. We present an improved approach using adaptive path sampling, iteratively reducing gaps between the base and target.
Using this adaptive strategy,  we develop two metastable sampling schemes. They are automated in Stan and require little  tuning. For a multimodal posterior density, we equip simulated tempering with a continuous temperature. For a funnel-shaped entropic barrier, we adaptively increase mass in bottleneck regions to 
form an implicit divide-and-conquer. 
Both approaches empirically perform better than existing methods for sampling from   metastable distributions, including higher accuracy and computation efficiency.

 \end{abstract}

\section {The normalizing constant and posterior metastability}

In Bayesian computation, the posterior distribution is often available as an \emph{unnormalized density} $q(\theta)$. The unknown and often analytically-intractable integral  $ \int\! q(\theta) d\theta$ is called the \emph{normalizing constant} of $q$.
Many statistical problems involve estimating the normalizing constant, or the ratios of them among several densities. For example, the marginal likelihood of a statistical model with likelihood $p(y|\theta)$ and prior $p(\theta)$ is the normalizing constant of $p(y|\theta)p(\theta)$:  $p(y)= \int\! p (\theta, y) d\theta$. The Bayes factor of two models $p(y|\theta_1), p(\theta_1)$ and $p(y|\theta_2), p(\theta_2)$, requires the ratio of the normalizing constants in densities $p(y, \theta_1)$ and $p(y, \theta_2)$.

Besides, we are often interested in the normalizing constant as a function of parameters. In a posterior density $p(\theta_1, \theta_2,\ldots, \theta_d | y)$, the marginal density of coordinate $\theta_1$ is proportional to $\int\!\ldots\int\! p(\theta_1, \theta_2,\ldots, \theta_d | y) d \theta_2,\ldots, d\theta_d$,  the normalizing constant of the posterior density with respect to all remaining parameters.  Accurate normalizing constant estimation means we have well explored region containing most of the posterior mass, which implies that we can then accurately also estimate posterior expectations of many other functionals.

In simulated tempering and annealing, we augment the distribution $q(\theta)$ with an inverse temperature $\lambda$ and sample from $p(\theta, \lambda) \propto q(\theta)^\lambda$. Successful tempering requires fully exploring the space of $\lambda$, which in turn requires evaluation of the normalizing constant as a function of $\lambda$: $z(\lambda)= \int\!  q(\theta)^\lambda d \theta$. Similar tasks arise for model selection and  averaging on a series of statistical models indexed by a continuous tuning parameter $\lambda$: $p(\theta, y | \lambda )$.  In cross validation, we attach to each data point $y_i$ a $\lambda_i$ and augment the model $p(y_i|\theta) p(\theta)$ to be $q(\lambda, y, \theta)=  \prod_i p(y_i|\theta)^{\lambda_i} p(\theta)$, such that the pointwise leave-one-out log predictive density  $\log p(y_i | y_{-i})$ becomes  the log normalizing constant $\log \int\! p(y_i|\theta)p\left(\theta\mid y, \lambda_i=0, \lambda_j=1, \forall j\neq i\right) d \theta$.


In all of these problems we are given an unnormalized density $q(\theta, \lambda)$, where $\theta \in \Theta$ is a multidimensional sampling parameter and  $\lambda \in \Lambda$ is a free parameter, 
and we need to evaluate the integrals at any $\lambda$,
\begin{equation}\label{eq_integral_target}
z: \Lambda \to \R, ~  z(\lambda)=  \int_\Theta  q(\theta, \lambda) d \theta, ~\forall  \lambda \in \Lambda.
\end{equation}
$z(\cdot)$ is a function of $\lambda$. For convenience, throughout the paper we will call $z(\cdot)$ the  \emph{normalizing constant} without describing it is a function.  We also use notations $q$ and $p$ to distinguish unnormalized and normalized  densities.

In most applications, it is enough to capture $z(\lambda)$ up to a multiplicative factor that is free of $\lambda$,  or equivalently the ratios of this integral with respect to a fixed reference point $\lambda_0$ over any $\lambda$:
\begin{equation}\label{eq_integral_ratio}
\tilde z(\lambda)=  z(\lambda) /z(\lambda_0),~\forall  \lambda \in \Lambda.
\end{equation}
\subsection{Easy to find an estimate, prone to extrapolation}\label{sec_why}
Two accessible but conceptually orthogonal approaches stand out for the computation of \eqref{eq_integral_target} and \eqref{eq_integral_ratio}. Viewing \eqref{eq_integral_target} as the expectation with respect to the conditional density $\theta|\lambda \propto  q(\theta, \lambda)$, we can numerically integrate \eqref{eq_integral_target} using quadrature, where the simplest is linearly interpolation, and the log ratio in \eqref{eq_integral_ratio} can be computed from first order Taylor series expansion,
\begin{equation}\label{eq_score}
\log \frac{z(\lambda)} {z(\lambda_0)} \approx (\lambda  - \lambda_0) \frac{d }{d\lambda }  \log z(\lambda)\vert_{\lambda=\lambda_0}   \approx (\lambda  - \lambda_0) \frac{1}{ z(\lambda_0)}\int_\Theta   \left(\frac {d}{d \lambda }q(\theta, \lambda) \vert_{\lambda=\lambda_0} \right)d \theta.
\end{equation} 

In contrast, we can sample from the conditional density $\theta| \lambda_0 \propto q(\theta, \lambda_0)$, and compute  \eqref{eq_integral_target}  by 
importance sampling, 
\begin{equation}\label{eq_importance_sampling}
 \frac{z(\lambda)}{z(\lambda_0)} \approx  \frac{1}{S}\ {\sum_{s=1}^S  \frac{  q (\theta_s,   \lambda) } {  q (\theta_s,  \lambda_0)}}, \quad \theta_{s=1, \cdots, S} \sim q (\theta,   \lambda_0).
\end{equation} 

How should we choose between  estimates \eqref{eq_score} and \eqref{eq_importance_sampling}? There is no definite answer. For example, in variational Bayes with model parameter $\theta$ and variational parameter $\lambda$, the gradient part of the \eqref{eq_score}  is  the \emph{score function} estimator based on the \emph{log-derivative trick}, while the gradient of \eqref{eq_importance_sampling} under a location-scale family  in  $q (\theta_s| \lambda)$ is called the Monte Carlo gradient estimator using the \emph{reparametrization trick}---but it is in general unknown which is better.   

On the other hand, both \eqref{eq_score} and \eqref{eq_importance_sampling} impose severe scalability limitations on the dimension of $\lambda$ and  $\theta$.  Essentially they \emph{extrapolate} either from the density conditional on $\lambda_0$ or from simulation draws from $q(\theta|\lambda_0)$ to make inferences about the density conditional on $\lambda$, hence depending crucially on how close the two conditional distribution $\theta|\lambda_0$ and   $\theta|\lambda$ are. Even under a normal approximation, the Kullback-Leibler (KL) divergence between these densities  scales linearly with the  dimension of $\theta$. Thus it takes at least $\mathcal{O}\left(\exp(\mathrm{dim}(\theta)\right)$ posterior draws to make \eqref{eq_importance_sampling} practically accurate.   The reliability of  \eqref{eq_score} is further contingent on  how flat the Hessian of $z(\lambda_0)$ is, which is more intractable to estimate. In practice, these methods can fail silently especially when implemented as a black-box step in a large algorithm without diagnostics. 

\subsection{Free energy and sampling metastability}
From the Bayesian computation perspective, the log normalizing constant is interpreted as the analogy of ``free energy'' in physics (up to a multiplicative constant), in line with  the interpretation of $\log q(\theta)$ to be the ``potential energy" in Hamiltonian Monte Carlo sampling. 

A potential energy is called metastable if the corresponding probability measure has some regions of high probability, but separated by low probability connections.
Following are two types of metastability, which both cause Markov chain Monte Carlo algorithms difficulty moving between regions. This pathology is often identifiable by a large $\hat R$ between chains  and low effective sample size in generic sampling \citep{vehtari2019rank}.

\begin{enumerate}
\item In an energetic barrier, the density is isolated in multiple modes,  and the transition probability is low between modes. In particular, the Hamiltonian---the sum of the potential and kinetic energy---is preserved in every Hamiltonian Monte Carlo trajectory. Hence, a transition across modes is unlikely unless the kinetic energy is exceptionally large.

\item In an entropic barrier, or funnel, the typical set is connected only through narrow and possibly twisted ridges. This barrier is amplified when the dimension of $\theta$ increases, in a way that a random  walk in high dimensional space can hardly find the correct direction.   
\end{enumerate}

 \begin{figure} 
	\includegraphics[width=\textwidth] {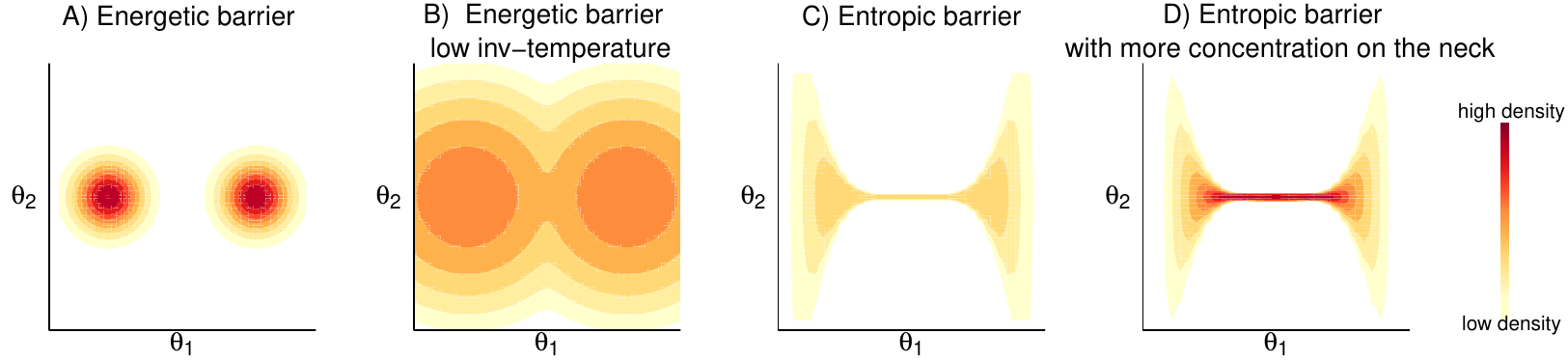}\caption{  \em An illustration of  metastability in a bivariate a distribution $p(\theta_1, \theta_2)$.  (A) With a  mixture of two Gaussian distributions, the energetic barrier prevents rapid mixing between modes.  (B) With a lower inverse temperature $\lambda$, the energetic barrier becomes flatter in the conditional distribution $p(\theta| \lambda)$. (C) With an entropic barrier, the left and right part of the distribution is only connected though a narrow tunnel, where the Markov chain will behave like a random walk. (D) Adding more density on the  neck increases the transition probability, while leaving $p(\theta_2|\theta_1)$ invariant. } \label{barrier}
\end{figure}

As the name suggests, we have the freedom to adaptively tune the free energy of the sampling distribution to remove the metastability therein.  The basic strategy is to augment the original metastable density $q(\theta)$ with an auxiliary variable $\lambda$, obtaining some density on the extended space $q(\theta, \lambda)$. For energetic barriers, we can take $\lambda$ to be an inverse temperature variable, a power transformation of the posterior density. The energetic barrier is flattened by a lower inverse temperature.   Temperature-based approaches cannot eliminate entropic barriers in the same way, but the transition is boosted by adding probability mass to the neck region.  Figure \ref{barrier} gives a graphical illustration.  

While the normalizing constant is not itself statistically meaningful in sampling from the augmented density $q(\theta, \lambda)$,  computing it serves as an essential intermediate step to constructing the otherwise intractable joint density. 
Unlike in usual Monte Carlo methods where the target distribution is given and and static, here as in an augmented system, we have the flexibility to either sample $\theta$ from some conditional distribution $\theta| \lambda$ at a discrete sequence of $\lambda$, or from some joint distribution of $(\theta, \lambda)$, treating $\lambda$ continuously.  
While continuous joint sampling has a finer-grained expressiveness for approximating the normalizing constant, it is harder to access the conditional distribution, which is ultimately what we need when $\lambda$ is an augmented  parameter. 

To facilitate the normalizing constant estimation and sampling in metastable distributions, 
the full problem contains three tasks.
\begin{enumerate}
    \item Determining what distribution we should sample $\theta$ and $\lambda$ from.
    \item Estimating the normalizing constants efficiently with the generated simulation draws. 
    \item Diagnosing the reliability of the sampling and estimation, particularly
    distinguishing between an informative extrapolation and a noisy random guess, and deciding when and where to adaptively resample.   
\end{enumerate}

In Section \ref{sec_method}, we introduce a practical solution to all three problems.  It extends the idea of path sampling \citep{Gelman1998Simulating} to an adaptive design, which performs both the  continuously-ranged  normalizing constant estimation, and direct sampling of the conditional density. Applying this strategy to metastable sampling,  we demonstrate in Sections \ref{sec_tempering} and \ref{sec_meta}  that our proposed adaptive path sampling method enables efficient sampling in both the energetic and entropic  bottlenecks,  and as a byproduct provides normalizing constant estimation and convergence diagnostics. In Section \ref{sec_review}, we   compare the proposed  adaptive sampling to other approaches and show that it is the infinitely dense limit of these  basic strategies  \eqref{eq_score} and \eqref{eq_importance_sampling}.
We  experimentally illustrate the advantage of the proposed methods in Section \ref{sec_exp}. For the purpose of log normalizing constant estimation and continuous tempering, we have automated our method in the general purpose software Stan \citep{stan}, and illustrate the practical implementation in the Appendix.

\section{Proposed method}\label{sec_method}
\subsection{The general framework of adaptive path sampling }\label{sec_method_general}
To begin with, we outline an adaptive path sampling algorithm for the general problem of normalizing constant (ratio) estimation \eqref{eq_integral_ratio} in a  $\lambda$-augmented system
$q(\theta, \lambda), \theta\in \Theta, \lambda \in \Lambda$.  We further elaborate its application in the context of metastable sampling in Section \ref{sec_tempering} and  \ref{sec_meta}. 

Instead of sampling $\lambda$ directly, we consider a transformed sampling parameter $a$ through $\lambda = f(a)$, where the link function
$f: \mathcal{A}\to \Lambda$ is continuously differentiable and $\mathcal{A}$ is the support of $a$.  For simplicity, we will use $\mathcal{A}= [0,1]$ in this section.   The actual sampling takes place in the  $\mathcal{A} \times \Theta$ space. 
If there is an interval $\mathcal{I}\subset\mathcal{A}$ which $f$ maps to a fixed value $\lambda_\mathcal{I}$, we can directly obtain conditional draws from $\theta|\lambda_\mathcal{I}$ by $\{\theta_i : a_i\in \mathcal{I}\}$, while not suffering from discretization errors of the the normalization constant $z(\lambda)= \int_{\Theta}q(\theta, \lambda)d\theta.$     We denote the conditional density $\pi_{\lambda}:= p(\theta|\lambda) = {q(\theta, \lambda) }/ z(\lambda)$.

The general algorithm then  iterates the following four steps.
\paragraph{Step 1. Joint sampling with invariant conditional densities.}
To start, we sample $S$ \emph{joint} simulation  draws $(\theta_i, a_i)_{i=1}^S$ from a joint density 
\begin{equation}\label{eq_joint_sample}
p(\theta, a)\propto  \frac{1}{c(\lambda)} q(\theta, \lambda), ~ \lambda= f(a),
\end{equation}
where $c(\lambda)$ is a parametric pseudo prior that is constructed using a series of kernels  $\{\gamma_i(\lambda)\}_{i=1}^I$ and regression coefficients $\{\beta_{cj}\}_{j=0}^J$ (which will be updated adaptively throughout the algorithm),
\begin{equation}\label{eq_parametric_form}
\log c(\lambda)= \beta_{c0} \lambda + \sum_{j=1}^I  \beta_{cj} \gamma_i(\lambda ). 
\end{equation}
By default, we initialize at a constant function $\beta_{c}=0$, i.e., $c(\lambda) \equiv 1$. No matter what the prior $c(\lambda)$ is, the conditional distributions  $\theta | a  \propto q(\theta, \lambda= f(a))$ in the joint simulation draws are invariant. This motivates to adaptively changing the pseudo-prior $c(\lambda)$.

For the joint sampling task \eqref{eq_joint_sample},  we will typically be using dynamic Hamiltonian Monte Carlo (HMC) \citep{hoffman:2014:nuts,Betancourt2017A} in Stan, which only requires the unnormalized log density  $\log q(\theta, \lambda) - \log {c(\lambda)}$ as input.

\paragraph{Step 2. Estimating the log normalizing constant from joint draws.}
Thermodynamic integration \citep[][see also Appendix \ref{sec_appendix_Thermodynamic}]{Gelman1998Simulating}  is based on the identity 
\begin{equation}\label{eq_gradient}
\frac{d}{da }\log z( f(a)) =  \E_{\theta| f(a)} \left(\frac{\partial}{\partial a} \log q(\theta, f(a))\right),
\end{equation}
where the expectation is over the invariant  conditional distribution $\theta | a  \propto q(\theta, f(a))$.

The ratio of the normalizing constant can be computed by integrating both sides of \eqref{eq_gradient}. To do this, we rank all the sampled draws according to their $a$ coordinate: $a_{(1)}<a_{(2)}< \cdots < a_{(s^*)}$, and compute the pointwise gradients  
\begin{equation}\label{eq:U_gradient}
U_{(i)}= \frac {\sum_{a_j=a_{(i)}} \frac{ \partial    }{\partial a} \log q(\theta, f(a))\bigr|_{\theta_j, a_j}} {\sum_{j: a_j=a_{(i)}}1  }.
\end{equation}
When there is no tie, the gradient estimate \eqref{eq:U_gradient} essentially  approximates the intractable pointwise integral in \eqref{eq_gradient}, 
$\E_{\theta| f(a_s)} \left( \frac{\partial}{\partial a} \log\left(q(\theta, f(a)\right)    \right)$ by \emph{one} Monte Carlo draw $\frac{\partial}{\partial a} \log\left(q(\theta, f(a)\right)|{\theta_s, a_s}$,  a common technique in stochastic approximation.

The integral of the right hand side of \eqref{eq_gradient} is then computed from these expectation estimates and the trapezoidal rule. For any $a^* \in \mathcal{A}$, we find its covering interval $a^* \in [a_{(i^*)}, a_{(i^*+1)})$, and compute its normalizing constant with reference to $z(f(0))$ by  
\begin{align}\label{path1}
 \begin{split}
 &\log \frac{ z( f(a^{*})) }{z(f(0))}  = \int_{0}^{a^{*}}\frac{d}{da }\log z(f(a)) da =  \int_{0}^{a^{*}} \E_{\theta| f(a)} \left(\frac{\partial}{\partial a} \log\left(q(\theta, f(a)\right) \right) da\\ 
 &\approx  \frac{1}{2} (a_{(1)}-0) (U_{(1)}  + U_{0} ) + 
 \frac{1}{2}\sum_{j=1}^{i^*-1} (a_{(j+1)}  - a_{(j)} ) (U_{(j+1)}  + U_{(j)}  ) +\frac{1}{2} (a^{*}  - a_{(i^*)} ) (U_{(i^*)} + U_{a^*}), 
  \end{split}
\end{align}
where $U_{a^*}$ and $U_{0}$ are obtained by extrapolating  $U$.

\paragraph{Step 3. Parametric regularization and adaptive updates.}
When the normalizing constant $z(\lambda)$ is only required up to a multiplicative factor, we can assume $z(f(0))=1$,  In Section \ref{sec_meta}, we show how to remove the fixed reference by additional self-normalization when the exact normalizing constant is needed.   

Equation \eqref{path1} yields an unbiased estimate of $\log z(\cdot)$. However, due to the stochastic approximation, \eqref{path1} has nonignorable variance in the region where not enough $a_i$ are sampled. For smoothness and regularization,  we approximate $\log z(\cdot)$ in some parametric family according to the ${L}_2$ distance criterion,  $\min \int_0^{1} \left(\log z(\lambda) - \left(\beta_0 \lambda + \sum_{j=1}^J  \beta_j  \gamma_j(\lambda )\right)\right)^2 da$. 
We compute this objective function on a uniform grid with length $I$: $\{\lambda^*_{i}= i/I,  1\leq i \leq I\}$, compute each $\log {z(\lambda^*_{i})}$ using estimate  \eqref{path1}, and solve the least squares regression 
\begin{equation}\label{eq_estimation_beta}
\beta_z= \arg\min_{\beta}\sum_{i=1}^{I}\left(\log z(\lambda^*_{i}) - \left(\beta_0 \lambda^*_{i} + \sum_{j=1}^J  \beta_j  \gamma_j(\lambda^*_{i}) )\right)\right)^2.
\end{equation}
This parametric estimation serves two goals. First, it provides us with a functional form for the prior which we  use in the next step to adaptively modify the sampling distribution. Second, the regression estimate \eqref{eq_estimation_beta} smooths  finite sample noise in \eqref{path1}, which is a bias-variance tradeoff.

We update the functional form of the pseudo-prior $\beta_c\coloneqq \beta_z$, or equivalently  $c(\cdot) \coloneqq z(\cdot)$.

\paragraph{Step 4. Diagnostics,  stopping condition, and mixing.}
The marginal distribution of $a$ from the sampling distribution \eqref{eq_joint_sample} satisfies $p(a)={z(\lambda)}/{c(\lambda)}, \lambda=f(a) $. If $z(\lambda)$ were accurately computed, one step adaptation  $c(a) \coloneqq z(a)$ would result in a uniform marginal distribution on $a$, which is the basis of diagnostics.  

Notably, the sampled marginal density $p(a)$ can be estimated as a normalizing constant   $p(a) = \int_{\Theta} q(\theta, f(a)) c^{-1}(f(a)) d \theta$, thus we use a similar  estimate as \eqref{path1}, only modifying the gradient $U_{(i)}$  by \begin{equation}\label{eq_gradient_marginal}
U_{(i)}^p= \frac  {\sum_{j: a_j=a_{(i)}} \frac{ \partial    }{\partial a} \Bigl(\log q(\theta, f(a)) -  \log c(f(a))\Bigr)\Bigl|_{\theta_j, a_j}} {\sum_{a_j=a_{(i)} }1}. 
\end{equation}

The sample estimates of $z(\lambda)$  and $p(\lambda)$ possess finite sample Monte Carlo error and are prone to over-extrapolation in regions of few $a$
draws. Therefore, we repeat Steps 1--3 until  $p(a)$ is ``functionally close enough" to a uniform density.
However, running until the complete uniformity is both in practice inefficient and  in theory unlikely to be obtained as the actual log normalizing constant $z()$ will not fall into the parametric family exactly.   

Our adaptation step  $z\to c$  can be viewed as an  importance sampling procedure from the joint proposal $c(f(a))^{-1}q(\theta, f(a))$ to the joint target $z(f(a))^{-1}q(\theta, f(a))$.
The importance ratio is  $r(a)=c(f(a))/z(f(a))=1/p(a)$, which only depends on the marginal of $a$, and the normalizing constant estimate \eqref{path1} can be equivalently expressed by the importance sampling estimate $z(\lambda)^{-1}=c(\lambda )^{-1}r(a)$ when the the marginal $p(a)$ is estimated from path sampling.  

To assess the accuracy of the final estimate, we use a Pareto-$\hat k$ diagnostic adapted from Pareto smoothed importance sampling \citep[PSIS,][]{vehtari2015pareto}. We fit the importance ratio $r_i= 1/p(a_i)$ in a generalized Pareto distribution, estimating its right tail shape parameter $\hat k$. As already applied in other computation diagnostics \citep[e.g.,][]{yao2018yes}, $\hat k$ quantifies the Renyi-divergence between the sampled density $c(a)^{-1}q(\theta, a)$ and the target   $z(a)^{-1}q(\theta, a)$ that has a uniform  marginal on $a$.  

When  $\hat k<0.7$, the normalizing constant $z(\lambda)$, viewed pointwise as an importance sampling estimate, is ensured to be reliable with a practical number of simulation draws, and we terminate sampling. The $\hat k$ threshold can be chosen smaller to make the decision more conservative. 

Otherwise, we perform further sampling with the updated pseudo prior $c$.   Crucially, the path sampling estimate \eqref{path1} is always unbiased for the log normalization constant under any sampling distribution as long as $\theta|a$ is left invariant. Thus, we save all previously sampled draws $\{a_s, \theta_s\}$, and mix them with the newly sampled draws 
in the normalizing constant estimation \eqref{path1} during each adaption.  
This remixing step corresponds to a divide-and-conquer strategy that we will further exploit to sample from a metastabe distribution with entropic barriers (Section \ref{sec_meta}).

\subsection{Adaptive continuous tempering: Sample from a multimodal distribution}\label{sec_tempering}
Given a statistical model, we can evaluate the posterior joint density  $p(\theta, y)$= prior$\times$likelihood.  In this section, we suppress the dependence on data $y$ and  denote the unnormalized posterior distribution from which we want to sample as $q(\theta) \coloneqq p(\theta, y), \theta\in \Theta$.  When  $q(\theta)$ exhibits severe multimodality, Markov chain Monte Carlo (MCMC) algorithms  have difficulty moving between modes. The state-of-the-art dynamic Hamiltonian Monte Carlo sampler for a bimodal density has a mixing rate as slow as random-walk Metropolis \citep{mangoubi2018does}, and even optimal tuning and Riemannian metrics do not help. Although it is possible to evade multimodal sampling using other post-processing and re-weighting strategies \citep[e.g.,][]{yao2020stacking}, we aim here to sample from the exact posterior density.  

To ease the energetic barrier between modes, we consider a distribution bridging between the target $q (\theta )$ and a base  distribution  $\psi(\theta)$  through a geometrically tempered  path:
$$
  p(\theta| \lambda)=  \frac{1}{z(\lambda)}  q (\theta )^{ \lambda}  \psi (\theta )^ {1- \lambda},   ~ \lambda \in [0,1],~ \theta\in \Theta, 
$$
where $\lambda$  is the augmented inverse temperature,  $z(\lambda)$ is the normalizing constant $z(\lambda) = \int_{\Theta}  \psi (\theta )^{ \lambda}  q (\theta )^ {1- \lambda}  d  \theta$, and
$\psi (\theta )$ is a proper base probability density, typically a simple initial guess or the prior that is easy to sample from. When $\psi$ is known exactly, $z(0)$ is 1.   We will discuss the choice of base distribution later.  $p(\theta| \lambda)$ is the target density when  $\lambda=1$, and becomes ``flattened" for a smaller $\lambda$.

 \begin{wrapfigure}{r}{7.6cm}
\vspace{-1em}
\includegraphics[width=\linewidth]{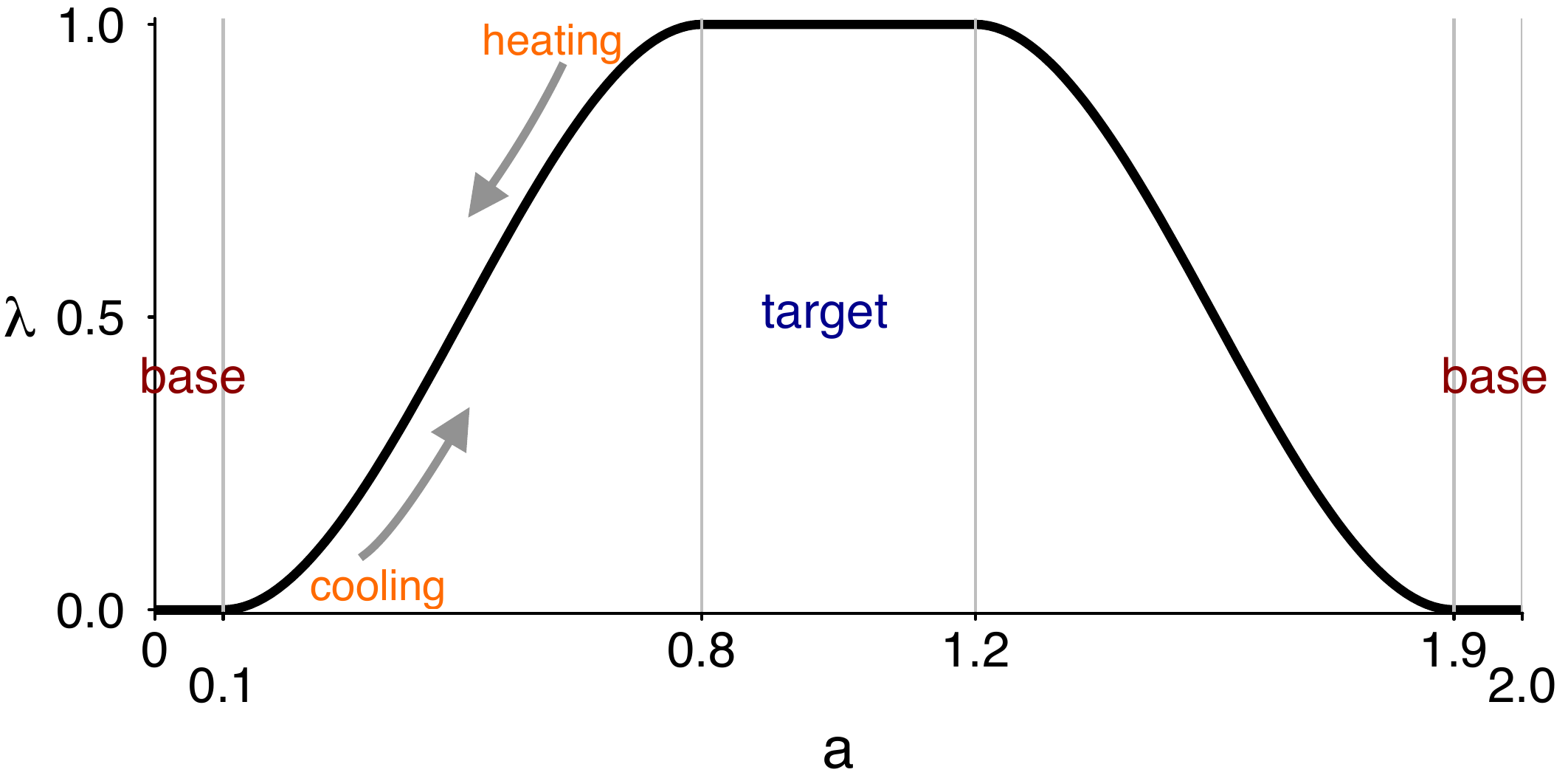}\vspace{-0.7em}\caption{\em   The link function $\lambda=f(a)$.  The flat area between 0.8 and 1.2 allows the continuous sampler to have a region where there are exact draws from the target distribution.} \label{fig_link}
\end{wrapfigure} 

To ensure that the joint sampler can access the
base and target distributions with nonzero probabilities, we define a link function $\lambda=f(a):  [0,2] \to [0,1]$, 
that is symmetric  $f(a)= f(2-a)$, and flat at two ends: $f(a)=0$ when $0\leq a \leq a_{\min} $  or $2-a_{\min} \leq a \leq 2$,   and $f(a)=1$ when $ a_{\max} \leq   a \leq 2- a_{\max}$.
This is easy to satisfy using a  piecewise polynomial, see  Figure \ref{fig_link} for an illustration.  In the experiment section we use $a_{\min}=0.1$, $a_{\max}=0.8$, and the concrete definition is in Appendix \ref{sec_appendix_link}.  

An $a$-trajectory from 0 to 2  corresponds to a complete $\lambda$ tour  from 0 to 1 (cooling) and back down to 0 (heating).  This formulation allows the sampler to cycle back and forth through the space of $\lambda$ continuously, while ensuring that some of the simulation draws (those with $a$ between 0.8 and 1.2) are drawn from the exact target distribution with $\lambda=1$.



To run simulated tempering, we apply the adaptive path sampling (Steps 1--4 in Section \ref{sec_method_general}) to the joint distribution,
$$p(\theta, a) \propto  \frac{1} {c(f(a))}  q (\theta )^{ f(a)}  \psi (\theta )^ {1- f(a)},  ~ a \in [0,2], ~\theta\in \Theta.    
$$
During each adaptation, we sample from this joint density, use all existing draws (including from previous adaptions) to obtain the path sampling estimated log normalization constant $\log z$, parametrically  regularize it, and update $\log c$ by $\log \hat z$.  Because $f()$ is designed symmetric, $z(f(a))= z(f(2-a))$, we flip all $a_s$ to be $2-a_s$ for all $a_{s}>2$ during log $z$ estimation \eqref{path1}. 

In addition, the pointwise gradient $U$ in \eqref{eq:U_gradient} is further simplified to be 
$$U_{(i)}= \frac {\sum_{j: a_j=a_{(i)}} f^{\prime}(a_{(i)})  \Bigl(\log q(\theta_{j})- \log \psi (\theta_j) \Bigr) }  {\sum_{j: a_j=a_{(i)}}1  }.
$$
If the base density $\psi (\theta)$ is chosen to be the prior in the model, this gradient is simply the product of $f^{\prime}(a_{(i)})$ and the log likelihood.

When the marginal distribution of $a$ is close to uniform, which is monitored by the Pareto-$\hat k$ diagnostic, a path has been constructed from the base to the target.  We then collect all draws  in the final adaptation  with temperature $\lambda = 1$, i.e. $\{\theta_i \mid f(a_i)=1\}.$  These are the desired draws from the target distribution $q(\theta)$. 
As a byproduct, we obtain the log normalization constant estimate $\log z$, and $z(1)$ equals the  marginal likelihood if $\psi$ is chosen as the prior. The full tempering method is summarized in Algorithm \ref{algorithm}.

\begin{algorithm}[!ht]
\KwIn{ $\psi (\theta), q (\theta) $: the base and (unnormalized) target density;} 
\KwOut{draws from the target distribution; 
$\log z(\cdot)$: log normalizing constant.} 
Initialize pseudo-prior $\log c(\cdot)=0$\; 
\Repeat
{
Pareto-$\hat k$, the estimated tail shape parameter of the ratios $1/p(a_s)$, is smaller than 0.7.
}
{ 
Sample $\{a_s, \theta_s\}_{s=1, \dots, S}$ from the joint density $q(\theta, a)= \frac{1}{c(a)}\psi(\theta)^{f(a)}q(\theta)^{(1- f(a))} $\;
Flip $a_s \coloneqq 2-a_s$ all $a_s>1$\;
Estimate $\log z(\cdot)$  by  path sampling (\ref{path1}), from draws in all adaptations\;
 Estimate $\log p(\cdot)$ by path sampling (\ref{path1}) and gradients \eqref{eq_gradient_marginal}, from the current adaptation\;
 Update $\log c(\cdot) \leftarrow \log z(\cdot)$\;
}
Collecting sample $\{\theta_i | f(a_i)=1\}$ as posterior draws of target distributions.
\caption{\em Continuous tempering with path sampling.} \label {algorithm}
\end{algorithm}

\subsection{Implicit divide and conquer in a metastable distribution}\label{sec_meta}
The proposed continuous simulated tempering algorithm in Section \ref{sec_tempering} alleviates metastablility in energetic barriers. However, tempering is not effective at overcoming purely entropic barriers. 
In such cases, instead of  augmenting the density with an additional temperature variable, we increase the density in the bottleneck region to encourage transitions between metastable regions. 

In many models, it is known that certain marginal distributions are problematic. For example, in hierarchical models,  
the centered parameterization effectively creates left truncation on  the group level standard deviation $\tau$, as the sampler hardly enters the $\tau\approx 0$ region.  We denote the joint distribution $q(\theta, \tau)$, where $\tau$ is the targeted problematic margin and $\theta$ is all remaining parameters.  In other cases, these problematic marginals can be identified by various MCMC diagnostics such as trace plots.  

A conservative solution is to sample $\tau$ first from some ``wider" proposal distribution,  then sample $\theta$ given $\tau$ in a Gibbs fashion, and finally adjust for the extra wide proposal by importance sampling. Here we propose an alternative strategy based on path sampling that does not require the Gibbs step or any closed form conditional density. The method is readily available using the joint sampler in Stan. 

We first sample \emph{some} simulation draws from the joint posterior distribution of all parameters $\tilde q(\theta, \tau)= q(\theta, \tau)$, without requiring a complete exploration.
The marginal density of $\tau$ in the  original model, $p(\tau)$,  is the normalization constant $p(\tau)\propto\int_{\Theta}q(\theta, \tau)d \theta.$ 
Hence, we compute $\log p(\tau)$ over all sampled $\tau$ using path sampling formula \eqref{path1}, and modify the sampling distribution by adding the bias term $\log \tilde q(\theta, \tau)\coloneqq \log  q(\theta, \tau)+ \log\left(p^{\mathrm{targ}}(\tau) / p(\tau)\right)$, where $p^{\mathrm{targ}}(\tau)$ is a desired marginal density. It can be fixed at any distribution that covers the true posterior marginal of $\tau$, such as its prior. 
we discuss other choices in the end of this section.  

We then sample new draws $(\theta, \tau)$ from this adapted  sampling distribution.
Since we only change the marginal distribution of $\tau$ between adaptations, the conditional densities $\theta|\tau$ remain invariant. We mix the new sample with the ones from all previous adaptations and use this cumulative sample to estimate the aggregated marginal density $p(\tau)$ using  path sampling at each adaptation.
We iterate the above procedure until we reach approximate  marginal convergence to $p^\mathrm{targ}(\tau)$, which is further quantified by the Pareto-$\hat k$ diagnostic.

$\tau$ is often undersampled in some regions in early iterations. In these regions, the path sampling estimated $p(\tau)$ is less reliable, and the  $\log\left(p^{\mathrm{targ}}(\tau) / p(\tau)\right)$ term will cause the sampler to focus in these undersampled regions on the next iteration (and mostly ignore those regions that were already sufficiently sampled). This adaptive behavior is what makes this algorithm an implicit divide-and-conquer-type procedure.  Algorithm \ref{alg:idc} shows the complete setup for this procedure.

Once we have obtained our estimate of $p(\tau)$ from the divide-and-conquer algorithm, we can use importance sampling to compute marginal posterior expectations and quantiles. Alternatively, the posterior expectation of any integrable $h(\tau)$  is the normalizing constant of $h(\tau)p(\tau)$, which we can evaluate using path sampling and (one-dimensional)  quadrature \eqref{path1}. In our experiments, we find the latter approach to be more robust. Furthermore, as a byproduct of our marginal density estimate, we can evaluate the marginal distribution function out to arbitrary distances. This allows us to estimate quantiles with extremely small tail probability that may be more difficult to estimate with Monte Carlo draws.

In addition, the path sampling estimate of $p(\tau)$ can be used to diagnose poor sampling behavior in standard HMC by comparing this estimate with the obtained empirical distribution.



\begin{algorithm}
\KwIn{$q(\tau, \theta)$:  the (unnormalized) joint density;
$\tau$: the problematic  margin; $\theta$: all remaining parameters;
$p^{\mathrm{targ}}(\tau)$: the targeted marginal of $\tau$.  } 
\KwOut{$p(\tau)$: marginal density of  $\tau$ in the original joint density;  
} 
Initialize sampling distribution $\tilde q =q(\theta, \tau)$; $j$=1\;
\Repeat
{ 
The ratios $r=\{p_j(\tau)/p_{j-1}(\tau) \}$  have $\hat k<$  0.7.
} 
{ 
Generate sample $\{\tau_{s}, \theta_{s}\}$ from $\tilde q(\tau, \theta)$ \;
Mix these draws with all previous adaptations $\{\tau_s, \theta_s\}_{s=1}^S$\;
Compute  $p (\tau)$, the  marginal density of $q(\tau,\theta)$,  using path sampling \eqref{path1}, gradients \eqref{eq_gradient_marginal} and all draws\;
Smooth estimated  $p(\tau)$ by regression \eqref{eq_estimation_beta}, and  record $p_j (\cdot)\coloneqq p(\cdot)$\;
Update sampling density $\tilde q(\tau, \theta)\coloneqq  q(\tau, \theta) p^{\mathrm{targ}}(\tau) / p (\tau)  $\; 
$j\coloneqq j+1$\;
}
\caption{\em Implicit divide-and-conquer scheme for metastable distributions}  \label{alg:idc}
\end{algorithm}

\paragraph{The optimal marginal distribution of $a$.}
Lastly, the adaptive  path sample estimate  does not depend on the marginal distribution of $a$ and $\lambda$,  so that this marginal  distribution is determined by user specification. By default in \eqref{path1} we use the update rule $z(\cdot) \to c(\cdot)$, which enforces a uniform  marginal distribution on $a$.  More generally, by updating
$c(f(\cdot)) \leftarrow z(f(\cdot))/p^{\mathrm{prior}}(\cdot)$, the final marginal distribution of $a$ will approach $p^{\mathrm{prior}}$.  The choice of $p^{\mathrm{prior}}$  is subject to a \emph{efficiency}-\emph{robustness} trade-off.  \citet{Gelman1998Simulating}  showed that the generalized Jeffreys prior 
$
p^{\mathrm{opt}}(\lambda) \propto \sqrt{\E_{\theta|\lambda}U^2(\theta,\lambda)}
$   minimizes the variance of the estimated log normalizing constant, where $U(\theta,\lambda)= \frac{\partial}{\partial \lambda} \log q(\theta, \lambda)$.
 With a slight twist,  in continuous tempering (Section \ref{sec_tempering}),  we can prove that another optimal prior
$
p^{\mathrm{opt}}(a)   \propto \frac{1}{{f'(a)}} \sqrt{ \mathrm{Var}_{\theta\sim p(\theta|a)}  U(\theta, a) }
$ ensures a smooth KL gap between two adjacent  tempered distribution in posterior sample  $\mathrm{KL}\bigl( \pi_a, \pi_{a+\delta a}  \bigr)\approx \mathrm{constant}$ for $a_{\min}< a <a_{\max}$ (Appendix \ref{sec_prior}). 
In discrete tempering, this constant KL gap is related to a constant acceptance rate in the neighboring Gibbs update. However, due to its dependence on the unknown normalizing constant (and higher orders),  these two efficiency-optimal priors require additional tuning and adaptations.  In continuous tempering, we  prefer the simple uniform  $a$  margin  for robustness, as it guarantees a complete $\lambda$ tour in the joint path.   In implicit divide and conquer, if the  efficiency is more of a concern, we recommend to adaptively update the target marginal to match the efficiency-optimal one $
p^{\mathrm{targ}}(\tau) \leftarrow p^{\mathrm{opt }}(\tau)  \propto \sqrt{ \E_{\theta\sim q(\theta,\tau)}  (\frac{\partial}{\partial \tau} \log q(\theta, \tau))^2}$, which can be further stochastically  approximated  by joint  Monte Carlo draws.  This additional adaptation is optional and the estimation of $p^{\mathrm{opt}}$ is not required to be precise.

\section{Related work: From importance sampling to adaptive importance sampling to Rao-Blackwellization to path sampling to adaptive path sampling}\label{sec_review}
There is a large literature on methods for computing or approximating normalizing constants that cannot be evaluated in closed form. We refer to \citet{Gelman1998Simulating} and  \citet{Lelievre2010Free} for comprehensive reviews on normalizing constants.

\paragraph{Adaptive importance sampling.}
 The basic importance sampling strategy \eqref{eq_importance_sampling} treats the normalizing constant $z(\lambda)$ as a conditional  expectation with respect to the random variable $\theta$, whose distribution is parameterized by $\lambda$.  
 \citet{chatterjee2018sample} proved that under certain conditions that the number of simulation draws required for the importance sampling estimate \eqref{eq_importance_sampling} of $z(\lambda)$ to have small error with high probability is roughly $\exp(\mathrm{KL}(\pi_\lambda ||\pi_{\lambda_0}))$.
 
 When this KL gap is too large, a remedy is to add more \emph{discrete} ladders  $\lambda_0<\lambda_1 <\ldots  <  \lambda_K=\lambda$, and use adaptive importance sampling.  At the $(\!j\!+\!1\!)$-th time,  we sample $\theta_{j1, \dots, jS}$ from $\pi_{\lambda_j}$, and the importance sampling estimate gives $z_{\lambda_{j+1}}/z_{\lambda_{j}}= 1/S \sum_{i=1}^{S} q (\theta_{j,i}, \lambda_{j+1})/ q (\theta_{ji}, \lambda_{j}).$ The final estimation of  normalizing constant is 
\begin{equation}\label{eq_two_point_ip}
\frac{z_{\lambda_K}}{z_{\lambda_0}}  =  \prod_{j=0}^{k-1} \frac{   z_{\lambda_{j+1}} } {   z_{\lambda_{j}}  } \approx \prod_{j=0}^{k-1} \left(1/S \sum_{i=1}^{S}  q (\theta_{ji}, \lambda_{j+1})/ q (\theta_{ji}, \lambda_{j})\right).
\end{equation}  

\paragraph{Bridge sampling.}
The importance sampling is restricted to sampling from the conditional density $\theta|\lambda$ for a constant $\lambda$ at each run---a slice in the $(\theta, \lambda)$ joint space. 
Bridge sampling simultaneously draws $\theta_{j1, \dots, jS_{j}}$ from $\pi_{\lambda_{j}}$ and  $\theta_{(j+1)1, \dots, (j+1)S_{j+1}}$ from  $\pi_{\lambda_{j+1}}$. 
Given a sequence of  integrable function $\{\alpha_{j}(\cdot)\}_{j=1}^{J}$,  we estimate the ratio of normalizing constant via  the bridge sampling    
  \citep{meng1996simulating} estimate:  

\begin{equation}\label{eq_bridge}
\frac{z_{\lambda_{j+1}}} {z_{\lambda_{j}}} = \frac{\E_{\pi_{\lambda_{j}}}    \left( \alpha_{j}(\theta) q(\theta, \lambda_{j+1})   \right)  } { \E_{\pi_{\lambda_{j+1}}}   \left(  \alpha_{j}(\theta)    q(\theta, \lambda_{j})    \right)  } \approx 
 \frac{ \sum_{i=1}^{S_j}  q (\theta_{ji}, \lambda_{j+1})  \alpha_{j} (\theta_{ji}) / S_j   }   {   \sum_{i=1}^{ S_{j+1}    }  q(\theta_{(j+1)i}, \lambda_{j})  \alpha_{j} (\theta_{(j+1)i})    / S_{j+1}}.   
\end{equation}


Under some independence assumptions,  the optimal choice of $\alpha_{j}(\theta)$ to minimize the variance of multistage bridge sampling estimate \eqref{eq_bridge} is 
$
\alpha^{\mathrm{opt}}_{j}(\cdot)= \frac{n_j  z_j^{-1}   }{ \sum_m  n_m z_m^{-1} q_m (\cdot)}
$ \citep{meng1996simulating, shirts2008statistically}.

\paragraph{Rao-Blackwellization.}
When some $\lambda_k$ are rarely seen, the inverse probability weighting is unstable.  Motivated by the identity $p(\lambda)= \E_{\theta}  p(\lambda |\theta)$,
\cite{Carlson2016Partition}  proposed a Rao-Blackwellized \citep{robert2013monte} estimate of the normalizing constant, 
essentially replacing the empirical marginal probability mass function $\mathrm{Pr}(\lambda)$ by a Rao-Blackwellized estimate 
$\mathrm{Pr}(\lambda=\lambda_k)= \sum_{s=1}^S \left(q(\theta_s, \lambda_k)/ \sum_{k'} q(\theta_s, \lambda_{k'})  \right)$.
We show in the appendix that this estimate is equivalent to multistate bridge sampling \eqref{eq_bridge} by taking
an empirical approximation of the theoretical optimum $\alpha^\mathrm{opt}_j$ defined above. 

\paragraph{Non-equilibrium methods.}
The importance sampling and bridge sampling estimates requires $n_j >1$ simulation draws from each $\theta|\lambda_j$. In non-equilibrium methods, we start by a simulation draw $\theta_0$ from equilibrium $\theta|\lambda_0$, evolve it through a sequence of transitions that keeps $\pi_k, k=1,\dots,K$ invariant at each step, and collect one non-equilibrium trajectory $(\theta_0,\dots,  \theta_{K-1})$. Notably, we do not draw from $\pi_K$ directly, and $\theta_j, j\geq 1$ is in general not $\pi_j$ distributed. We still obtain an unbiased estimate \citep{jarzynski1997nonequilibrium, neal2001annealed}: 
\begin{equation}\label{eq_jarzynski}
\frac{z_{\lambda_K}}{z_{\lambda_0}} = \E\left(\exp \mathcal{W}(\theta_0, \dots, \theta_{K-1})\right), \quad
\mathcal{W}(\theta_0, \dots, \theta_{K-1})=\log  \prod_{j=0}^{k-1}\frac{q(\theta_j, \lambda_{j+1})}{q(\theta_j, \lambda_{j})}.    
\end{equation}
where the expectation is over all initial draws and trajectories.

\paragraph{Thermodynamic integration.}
The path sampling estimate \eqref{path1} is unbiased for the log normalizing constant (ratios) $\log z$, while other importance sampling based algorithms are unbiased in the scale of the normalizing constant $z$.  In our adaptive procedure, since we update the \emph{logarithm} of the joint density and compute its gradient in the next Hamiltonian Monte Carlo run, the unbiasedness of $\log z$ is more relevant. In statistical physics, the $\mathcal{W}$ quantity in \eqref{eq_jarzynski} is interpreted  as virtual work  induced on the system. Jensen's inequality leads to $\E \mathcal{W} \geq \log z(\lambda_K)- \log z(\lambda_0)$. This is a microscopic analogy of the second law of thermodynamics: the work entered in the system is always larger than the free energy change, unless the switching is processed infinitely slow, which corresponds to the  thermodynamic integration.

The thermodynamic integration equality \eqref{eq_gradient} was first introduced by \citet{kirkwood1935statistical} in statistical physics, and further refined or applied by \citet*{ogata1989monte, neal1993probabilistic, Gelman1998Simulating, rischard2018unbiased} in the context of normalizing constant computing. However, the typical use of thermodynamic integration requires discretizing $\lambda$ into a fixed quadrature ladder $\lambda_1<\ldots< \lambda_K$, on which the gradient $U_j$ is computed using many draws from $\theta | \lambda_j$.  
These ladders involve further manual tuning \citep{schlitter1991methods, blondel2004ensemble} to control the variance of $U$, and the discretization error (bias) in the numerical integration is non-vanishing in a finite discrete ladder, to a large extent compromising the unbiasedness property of $\log z$ estimation.   

Our present paper also extends the general discussion of path  sampling in \citet{Gelman1998Simulating}, with added steps on parametric regularization, iterative adaptations, and diagnostics. Essentially we use stochastic approximation \citep{robbins1951stochastic} to compute the the pointwise gradient \eqref{eq_gradient}.
We employ a carefully designed link function $f$ that allows direct access to simulation draws from some chosen conditional distributions $\theta|\lambda$, while also eliminates the discretization bias. 
These extensions  facilitate the continuous tempering scheme by providing direct access to draws from the target distribution.

\paragraph{Path sampling as the continuous limit of importance sampling and Taylor expansion.} 
In Section \ref{sec_why}, we describe two separate approaches: the importance sampling estimate \eqref{eq_importance_sampling} and  Taylor series expansion \eqref{eq_score}. They reach the same first order limit when the proposal is infinitely close to the target. That is,  for any fixed $\lambda_0$, as $\delta=|\lambda_1 - \lambda_0|\to 0$,
$$\frac{1}{\delta}\log  \E_{\lambda_{0}}\left(\frac{q(\theta|\lambda_{1})}{q(\theta|\lambda_0)}\right) =   \int_{\Theta} \frac{\partial}{\partial \lambda} \log q(\theta|\lambda_{0}) p(\theta|\lambda_0)d\theta    + o(1)=
\frac{1}{\delta} \E_{\lambda_{0}}\left( \log \frac{q(\theta|\lambda_{1})} {q(\theta|\lambda_0)}\right).$$
The path sampling estimate $\int_{\lambda_0}^{\lambda_1}  \int_{\Theta} \frac{\partial}{\partial \lambda} \log q(\theta|\lambda_{l}) p(\theta|\lambda_l)d\theta d   \lambda$ that we employ is the integral of the dominate term in the middle. In this sense, path sampling  is the continuous limit of both the importance sampling and the Taylor expansion approach. For a more rigorous derivation, see Appendix \ref{sec_limit}.    

More generally, thermodynamic integration \eqref{eq_gradient} can be viewed as the $K\to \infty$ limit of the equilibrium bridge sampling, Rao-Blackwellization, and annealed importance sampling, but without having to fit the conditional model infinitely many times  \citep[for rigorous proofs, see Appendix \ref{sec_rb_bri} and \ref{sec_limit};  see also discussions in][]{Gelman1998Simulating, Carlson2016Partition, Lelievre2010Free}. We will further elaborate in the simulating tempering context that such continuous extension is desired, as otherwise the necessary number of interpolating temperatures $K$ soon blows up when the dimension of $\theta$ increases.

\paragraph{Simulated tempering.}
Simulated tempering and its variants provide an accessible approach to sampling from  a multimodal distribution.  We augment the state space $\Theta$ with an auxiliary  inverse temperature parameter $\lambda$, and employ a sequence of interpolating densities, typically through a power transformation $p_j \propto  p(\theta|y)^{\lambda_j}$ on the ladder $0< \lambda_1<\ldots \lambda_K=1$, such that  $p_K$ is the distribution we want to sample from and $p_0$ is a (proper) base distribution. At a smaller $\lambda$, the between-mode energy barriers in $p(\theta | \lambda)$ collapse and the Markov chains are easier to mix. This dynamic makes the sampler more likely to fully explore the target distribution at $\lambda=1$.

Discrete simulated tempering \citep{marinari_simulated_1992, neal1993probabilistic, geyer1995annealing} samples from the joint distribution $p(\lambda, \theta) \propto  1/ c(\lambda)q(\theta)^\lambda$ using a Gibbs scheme, where  $c(\lambda)$ is a pseudo-prior that is often iteratively assigned to be  
$\hat z(\lambda)$: an estimate of the normalizing constant of $q(\theta)^\lambda$ which may be obtained by any of the methods discussed above.     Each Gibbs swap involves sampling  $\theta|\lambda$ with a one- or multi-step Metropolis update that keeps  $p(\theta|\lambda)$ invariant, and a random walk in $\lambda$ that leaves  $\lambda|\theta$ invariant.  The number of Metropolis updates, the number of temperature samples, and the temperature spacing all involve intensive user tuning.

In simulated annealing \citep{neal1993probabilistic, morris1998automated},  we evolve $\lambda$ through a fixed schedule, and  update $\theta$ using a Markov chain  targeting the  current conditional distribution $p(\theta | \lambda)$. The annealed importance sampling \citep{neal2001annealed}  adjusts the non-equilibrium in the $\theta | \lambda$ update by the importance weight defined in \eqref{eq_jarzynski}.

However, as we mentioned above, discrete tempering schemes are sensitive to the choice of  $\lambda$ ladders.
The Markov chain must usually proceed by making small changes between the neighboring distributions.  Following the discussion in \citet{betancourt_adiabatic_2014}, 
if some pair of $\pi_{\lambda_j}$ and $\pi_{\lambda_{j-1}}$ have a large KL divergence, the error in importance sampling \eqref{eq_two_point_ip} based algorithms will be dominated by the $j$-th term. 
Under the optimal design, the Kullback–Leibler (KL) divergence between neighboring distributions should be roughly constant:
\begin{equation}\label{eq_constant_gap}
 \mathrm{constant}~\approxeq ~\mathrm{KL}\Bigl( \pi_\lambda, \pi_{\lambda+\delta \lambda}  \Bigr) = \int \log \bigl( \frac{\pi_{\lambda}}{\pi_{\lambda+\delta \lambda}}   )d\pi_{\lambda}  
  =\log \bigl(\frac{z(\lambda+\delta \lambda)}{ z(\lambda)}\bigr)-( \delta \lambda)^2   \frac{d}{d\lambda}\log z(\lambda),
\end{equation}
which can hardly be achieved even adaptively due to the reliance on the unknown $\log z$ and its derivative.

Further, even with a constant KL gap, the discrete ladder imposes dimension limitations.  For example, in simulated tempering, a transition  $(\theta, \lambda_i )\to (\theta, \lambda_j)$ between two temperatures $\lambda_i$ and $\lambda_j$ in the Gibbs update would only be likely if there is significant overlap between the potential energy distributions $ \log q(\theta|\lambda_i)$ and $\log q(\theta|\lambda_j)$. Under a normal approximation with dim$(\Theta)=d$,  the width of the energy distribution scales by  $\mathcal{O}(d^{1/2})$, while the distance between two adjoining energy distributions is $\mathcal{O}(d/K)$. This leads to the best case bound on the necessary number of interpolating densities $K=\mathcal{O}(d^{1/2})$. In practice, $K$ is recommended to grow like $\mathcal{O}(d)$ \citep{madras2003swapping}. More theoretical discussion shows that $K= \mathcal{O}(d)$ is often needed to ensure a polynomial bound on the adjacent temperature overlap and rapid mixing \citep{woodard2009conditions}.
Since the update on $(\lambda_k)_{k=1}^K$ behaves as a random walk, even when the $\theta|\lambda$ update takes zero time,  the relaxation time of a diffusion process with discrete state space $(\lambda_k)_{k=1}^K$ often  scales by $\mathcal{O}(K^2)$, soon becoming unaffordable as $K$ grows.

 
\paragraph{Toward continuous tempering.}
\cite{gobbo2015extended} designed a continuous tempering scheme by  adding
a single auxiliary variable $a\in R$ to the system. The inverse temperature $f(a)$ is defined by a smooth function such that $f(a)=1$ for $|a|<\Delta$  and $f(a)=0.15$ for $|a|>\Delta^*$.  The Hamiltonian of the system is modified to be $\hat H(\theta, p, a, p_a)= f(a) H(\theta, p) +{ p_a^2}/{m_a}+  \log z(a),$ where $p_a$ is the momentum of $a$, and $\log z(a)$ is adaptively updated using importance sampling, in order to force $a$ to be  uniformly distributed in the interval $[\Delta^*, \Delta^*]$.  Similarly, \cite{betancourt_adiabatic_2014}  introduced adiabatic Monte Carlo, where a contact Hamiltonian is defined on the augmented Hamilton system  $\hat  H(\theta, p, \lambda)= -\lambda  \log q (\theta) + 1/2  p^T  M^{-1}  p +\log z(\lambda)$. These methods are shown to outperform discrete tempering, but they require a modified Hamiltonian and are sensitive to task-specific implementation and  tuning.  

 \cite{Graham2017Continuously} formulated a continuous tempering on the joint density $p(\theta, \lambda) \propto {\zeta^{-\lambda}}{   q(x)^\lambda  \psi (x)^{(1-\lambda)}}$.  
This can be viewed as a special case of our proposed method by restricting the parametric form of our normalizing constant estimate to a
single parameter exponential function
$z(\lambda)=\zeta^{\lambda}$.
This method does not directly acquire  simulation draws from $ \theta | \lambda=1$ or   $0$, and integrals under target distribution are evaluated  through importance weighting. 



 \section{Experiments}\label{sec_exp}
To manifest the advantage of the proposed method, we present a series of experiments. In Section \ref{sec_exp_beta}, we use a conjugate model to compare the accuracy of  normalizing constant estimation.
In  Section \ref{sec_exp_gaussian}, we show that the  continuous tempering with path sampling is scalable to high dimensional multimodal posteriors. 
In  Section \ref{sec_exp_ess} and \ref{sec:HS}, we highlight the computational efficiency from the proposed method, including higher effective sample size and quicker mixing time. Lastly, we validate the implicit divide and conquer in a funnel shaped posterior in Section \ref{sec_exp_idc}.
 
\subsection{Beta-binomial example:  comparing accuracy of estimates of the normalizing constant}\label{sec_exp_beta} 
  We start with an example adapted from \citet{betancourt_adiabatic_2014},
  in which the true normalizing constant can be evaluated analytically. Consider a model with a binomial likelihood $y \sim \mathrm{Binomial}(\theta, n)$ and a Beta prior $\theta \sim \mathrm{Beta}(\alpha,\beta)$. Along a geometric bridge between the prior and posterior, the conditional unnormalized tempered posterior density at  an  inverse temperature $\lambda$ is
  $$q(\theta| \lambda) = \mathrm{Binomial}(y |\theta, n) ^{\lambda} \mathrm{Beta} (\theta|\alpha, \beta), \quad \theta\in [0,1], ~ \lambda\in [0,1].$$
Due to conjugacy, the normalized distribution has a closed form expression $p(\theta | \lambda)= \mbox{Beta}(\lambda y+ \alpha, \lambda(n-y)+\beta),$ and the true normalizing constant $z$ is  $$z(\lambda)= \int_0^1 q(\theta | \lambda)d\theta   = \left(\frac{\Gamma(n+1)}{\Gamma(y+1)\Gamma(n-y+1) }\right)^\lambda   \frac{\Gamma(\alpha+\beta)}{\Gamma(\alpha)\Gamma(\beta)}\frac{ \Gamma(\lambda y+\alpha)  \Gamma(\lambda(n-y)+\beta) }{\Gamma(\lambda n +\alpha+\beta)}. $$
  
   \begin{figure} 
\includegraphics[width=\textwidth] {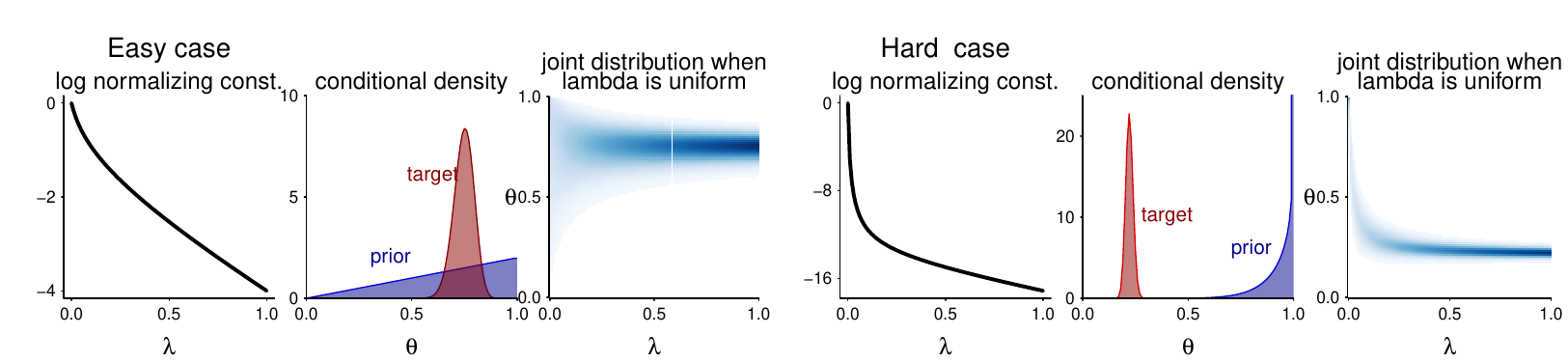}\caption{\em (a) The analytic log normalizing constant. (b) the base and the target (c) the joint of $(\lambda, \theta)$ when the marginal of $\lambda$ is uniform.    In  the hard case the target and base is separated and the log normalizing constant  changes rapidly. } \label{illus}
\end{figure}
  
  \begin{figure} 
 \includegraphics[width=\textwidth] {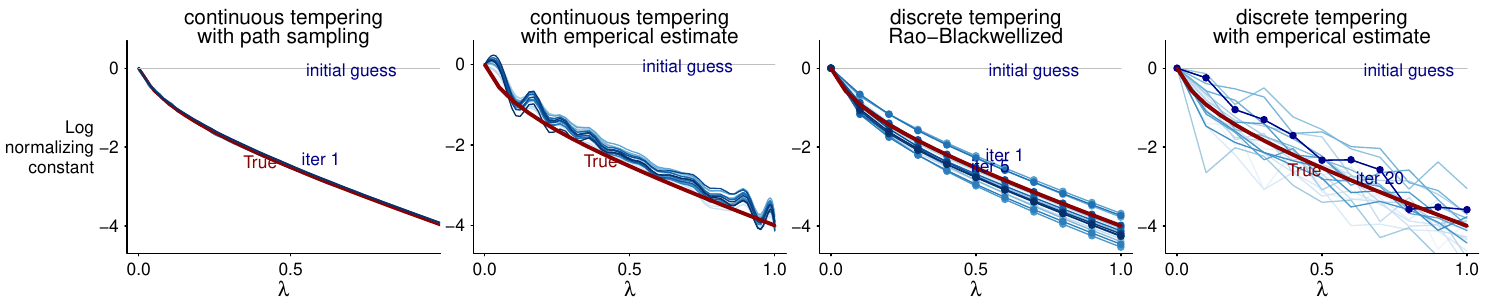}\caption{\em  Estimation of the log normalizing constant in the easy Beta-binomial example among first 20 adaptations.  All  methods eventually approximate the true function, while the proposed continuous tempering with path sampling has the fastest convergences rate.} \label{easy}
 \end{figure}
  
We compare four sample-based estimates of the log normalizing constant:
(i) continuous tempering with adaptive path sampling (the proposed method) with joint density $p(a, \theta) \propto 1/c(f(a)) q (\theta|f(a))$; 
 (ii) continuous tempering, but replacing the path sampling estimate of marginal density $p(a)$ by an
empirical estimate and using importance sampling $\hat z(a)= p(a)c(a)$ to compute and update $z$;   (iii) discrete  simulated tempering with importance sampling estimation and estimating the marginal probability mass function $\mathrm{Pr}(\lambda= \lambda_i), 0= \lambda_0 < \lambda_1<\ldots< \lambda_K=1$ by Monte Carlo average;  
(iv) discrete tempering with a Rao-Blackwellized \citep{Carlson2016Partition} estimate of the probability mass function.

In the first setting (left half  of Figure \ref{illus}), we set $\alpha=2$, $\beta=1$, $y=60$, $n=80$.   
Figure \ref{easy} presents the log normalizing constant estimates in the first 20 adaptations. All methods start with a flat initial guess $\log z(\lambda)=0$. In continuous tempering, we draw 3000 joint $(a, \theta)$ draws in each adaptation. The first half of the draws are treated as warm-up and discarded (which also do in the discrete setting). We choose a length $I=100$ approximation grid in the parametric adaptation  step \eqref{eq_parametric_form}.   For discrete tempering, we use an evenly spaced ladder $\lambda_i =(i-1)/10, i=1,\dots, 11$ and draw 150 $\lambda$ draws per adaptation, each followed 100 HMC updates in $\theta$. These numbers ensure the continuous tempering has a smaller computation cost (3000 joint draws) than in discrete implementations (100 HMC updates on $\theta$ $\times$ 150 updates on $\lambda_i$) per adaptation, so as to make the efficiency comparison convincing. 

   \begin{figure} 
 \includegraphics[width=\textwidth] {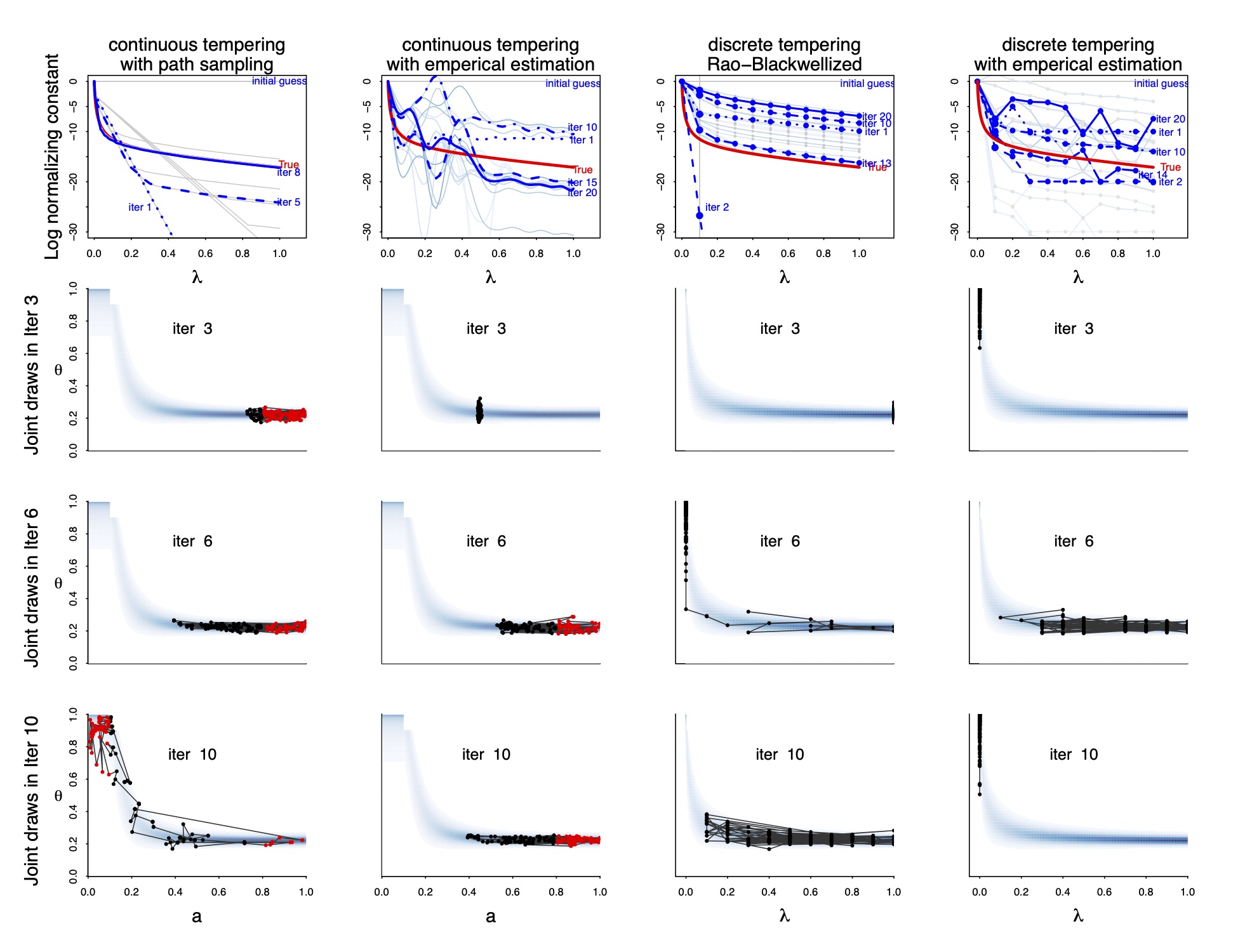}\caption{\em Comparison of four tempering methods in the hard beta-binomial example. Row 1: Starting from a flat guess, only the proposed method converges to the  the true (red curve) value of the log normalizing constant after 8 adaptations.   Rows 2--4 compare the first 150 draws of the joint distribution of parameters and temperatures in adaptations 3, 6, and 10.  The proposed method fully explores the the joint space efficiently, while the two discrete schemes exhibit random walk behaviour in temperatures updates, and cannot adapt to the rapid changing regions near $\lambda=0$.} \label{hard}
 \end{figure}

    \begin{figure} 
 \includegraphics[width=\textwidth] {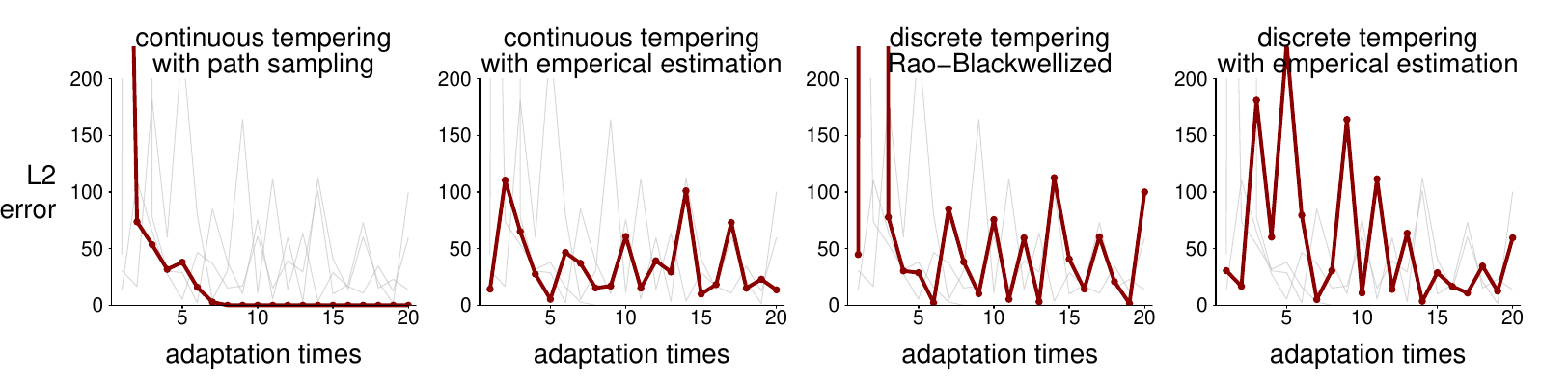}\caption{\em  $L_2$ errors of the log normalizing constant estimate $\log z(\lambda)$ from four methods during the first 20 adaptations.  Only continuous tempering with path sampling  gives a monotonically decreasing error which shrinks to zero in practical amount of time.} \label{L2}
 \end{figure}

In this easy case, all methods eventually approximate the truth(red curve), among which our proposed continuous tempering with path sampling  has the quickest convergence rate, almost recovering the truth within the first adaptation.

In the hard case (right half of Figure \ref{illus}), $\alpha=9$, $\beta=0.75$, $k=115$, $n=550$. The base and target are two isolated spikes, and the log normalizing constant changes rapidly especially for small $\lambda$ . 
 Figure \ref{hard} compares four log normalizing constant estimation methods  in the first 20 adaptations.  Continuous tempering with path sampling nearly converges to the true value after the 8th adaptation.  Rao-Blackwellization achieves a reasonable discrete approximation at the 13-th adaptation, but the result is unstable due to discretization error. In fact, Figure \ref{L2} displays the $L_2$ error of $\log z$. Only the proposed method has a  monotonically decreasing error with more adaptations. The two empirical importance sampling based methods are  both too noisy to be useful.  
 
 Rows 2--4 in  Figure \ref{hard}  show the first 150 joint draws in adaptation 3, 6, and 10 from all four methods. Our proposed method fully explores the joint space in adaptation 10. Here, the joint HMC jumps are gradient-guided and and make subtle jumps in regions of rapid change (where $\lambda\approx 0$). In contrast,  since the conditionals at $\theta| \lambda=0$ and $\lambda=0.1$ differ significantly, the discrete tempering procedure cannot automatically impute more ladders among them, and always over-weighs one end.

\subsection{Sampling from Gaussian mixtures: compare accuracy of the multimodal sampling}\label{sec_exp_gaussian}
Next we consider the problem of sampling from Gaussian mixtures. 
For visual illustration
we begin with the simple case of a mixture of two one-dimensional Gaussians.

Figure \ref{fig:GMM_1d} shows five out of ten total iterations for which we ran our tempering algorithm. At initialization, the joint sampler can only see a thin slice of $a$ in the region around the base, with a large resulting Pareto-$\hat k$ diagnostic. With more adaptations, more temperatures are sampled, and the path sampling estimate relies less on extrapolation. After four adaptations, the sampler has fully explored $a\in [0,2]$, and completed a base-target bridge in the augmented space.  Collecting simulation draws with $f(a)=1$ retrieves the target distribution. The accuracy is confirmed by the $\hat k<0.7$, or a visual check of a nearly uniform $a$ marginal distribution.

Next we consider a 10-component mixture of Gaussians.  We generate the individual Gaussian components with variance a tenth of the minimum distance between any two of the mode centers to ensures separation between the modes. 
We perform this sampling in a range of dimensions from 10 to 100 (with the number of components fixed).   We compare the proposed tempering with path sampling with two benchmark algorithms (i) Rao-Blackwellized discrete tempering \citep{Carlson2016Partition}, and  (ii) continuous tempering with a log linear normalizing constant assumption \citep{Graham2017Continuously}. We do not include other empirical importance sampling based tempering as they are strictly worse than these two benchmarks. 
For all methods, we use an over-dispersed normal base distribution.


\begin{figure}
    \centering
    \includegraphics[scale=0.65]{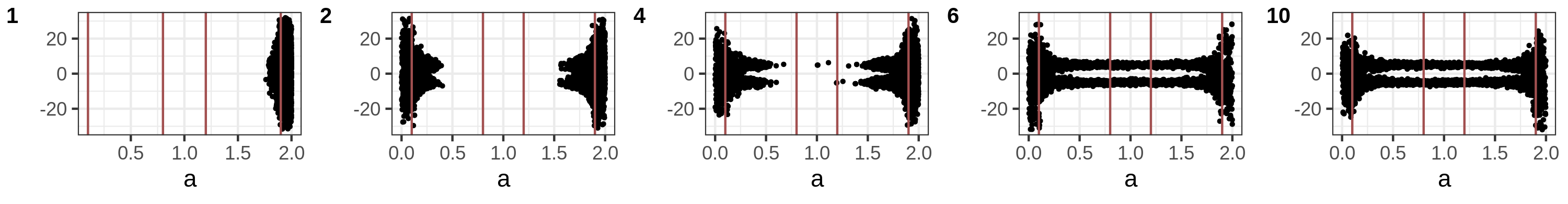}
    \includegraphics[scale=0.65]{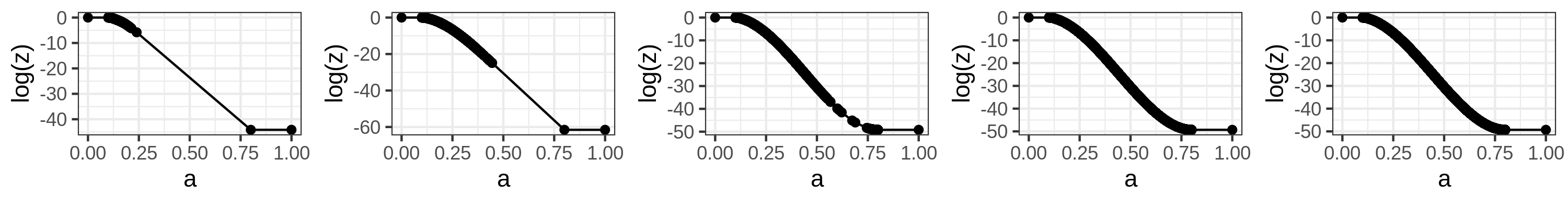}
    \includegraphics[scale=0.65]{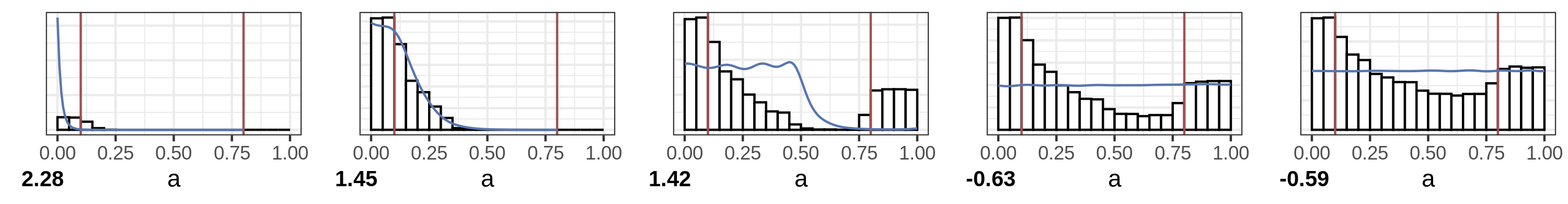}
    \includegraphics[scale=0.65]{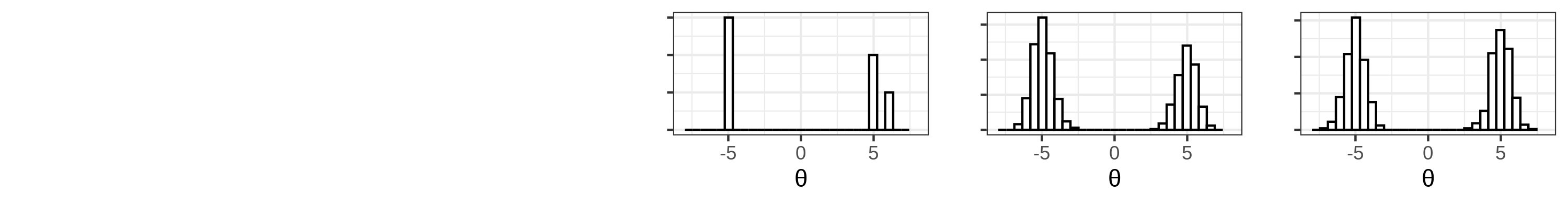}
    \caption{\em Path sampling-based tempering for the target a Gaussian mixture,  plotting adaptation 1, 2, 4, 6, and 10.  The first row shows the joint simulation draws;  the second row displays the estimate of the log normalizing constant; the third row displays the marginal density on the temperature and cumulative draws of the temperature variable, with Pareto-$\hat k$ diagnostics printed at the bottom left  of each panel (good if $\hat k <0.7$); and the fourth row is the draws from the target density. }
    \label{fig:GMM_1d}
\end{figure}
\begin{figure}
    \centering
    \includegraphics[scale=0.45]{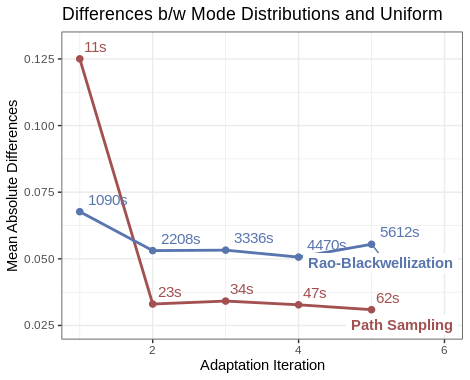}
    \includegraphics[scale=0.45]{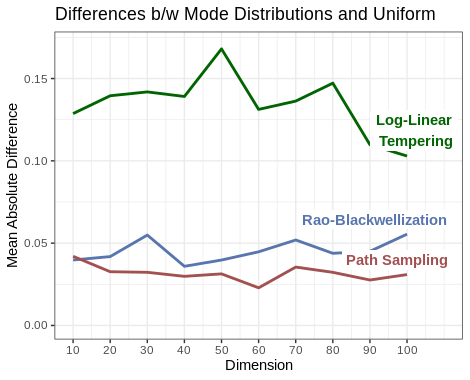}
    \includegraphics[scale=0.56]{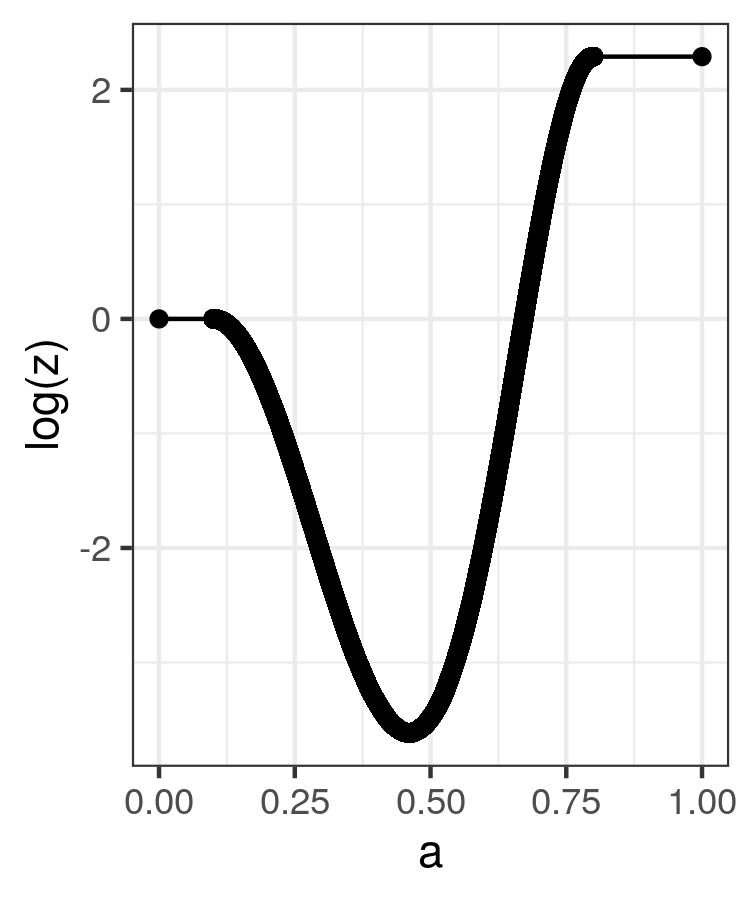}
    \caption{\em Mean absolute error of target simulation draws from adaptive path sampling based tempering (proposed), and two benchmarks: log-linear continuous tempering and Rao-Blackwellization.   Results are averages over 5 repeated runs.  Left: Sampling error and time at each adaptation for a Gaussian mixture target with 100 dimensions and 10 components. Numerical labels indicate the number of seconds of CPU time before each adaptation completed. Middle: Comparison of sampling errors as dimensions range from $d=10$ to $100$.  The log-linear tempering is supplied with true normalization constant while the other two are with uniform initialization. 
    Right: The estimated log normalizing constant log $z (f(a))$ from path sampling.     
  }
    \label{fig:GMM_example_2}
\end{figure}

In order to assess whether a given tempering procedure has succeeded, we count the proportion of draws in each target mixture component for each tempering method, and compare this distribution to the actual equal weight target.  The results are averaged over five independent runs.  
Figure \ref{fig:GMM_example_2} displays the evolution of 
the mean absolute difference between the sampled mode proportions in each adaptation of each algorithm and the target in the 100-dimensional case.
The left panel of Figure \ref{fig:GMM_example_2} labels each point with the total CPU time needed to complete each adaptation iteration. Not only does the path sampling algorithm achieve more uniform mode exploration, but it does so in significantly less computation time than the Rao-Blackwellized procedure.
This is partly a symptom of the HMC-in-Gibbs implementation of the Rao-Blackwellized algorithm. 

The middle panel of Figure \ref{fig:GMM_example_2} displays the same mean absolute difference between sampled and actual target distribution against the dimension
for continuous tempering with path sampling, the Rao-Blackwellized procedure, both with a uniform initialization, and the log linear tempering initialized at the true normalizing constant. For path sampling and Rao-Blackwellized sampling, we plot the results after five iterations of adaptation.  In all cases, the final in-target sample from the proposed method is closest to the actual target, with the the smallest mean absolute differences. As shown in the right panel, the normalization constant $\log z(f(a))$ is not monotone or linear, which explains the undestied performance of  the log linear tempering even when it is supplied with the ground truth normalizing constant.

In Figure \ref{fig:GMM_example_2}, we do not see the mode exploration degrade with increasing dimension for path sampling. This is because the log normalizing constant itself scales mildly with the dimension in Gaussian mixtures.  It general, the number of dimensions is not the limiting factor of path sampling, but it could inflate the  log normalizing constant in a severe base-target mismatch. We further discuss dimension scalability in Appendix C.

\subsection{Flower target: the gain on computation efficiency}\label{sec_exp_ess}

Next we consider sampling from a flower shaped distribution as previously used in \citet{sejdinovic2014kernel} and \citet{nemeth2019pseudo}. This is a two-dimensional distribution with probability density spread throughout multiple ``petals''. Its probability density function is given by
\[
p(x,y\mid \sigma,r,A,\omega) \propto \exp\left(-\frac{1}{2\sigma^2}\left(\sqrt{x^2+y^2} - r -A\cos(\omega\arctan(y,x))\right)\right).
\]
The probability mass is pinched into narrow regions between petals, slowing exploration of the target. For this reason, the flower distribution is challenging to sample from using standard HMC. 

Figure \ref{fig:flower_example_1} plots draws from the flower distribution with 6 petals using plain HMC and using continuous tempering with path sampling. As before, we use a normal base distribution. Both algorithms were run for enough time to generate a similar number of draws from the target. Standard HMC clearly fails to adequately explore each of the petals of the flower distribution. Path sampling-based continuous tempering succeeds in generating draws from each of the petals in roughly equal proportions.
\begin{figure}[!b]
    \centering
    \includegraphics[scale=0.33]{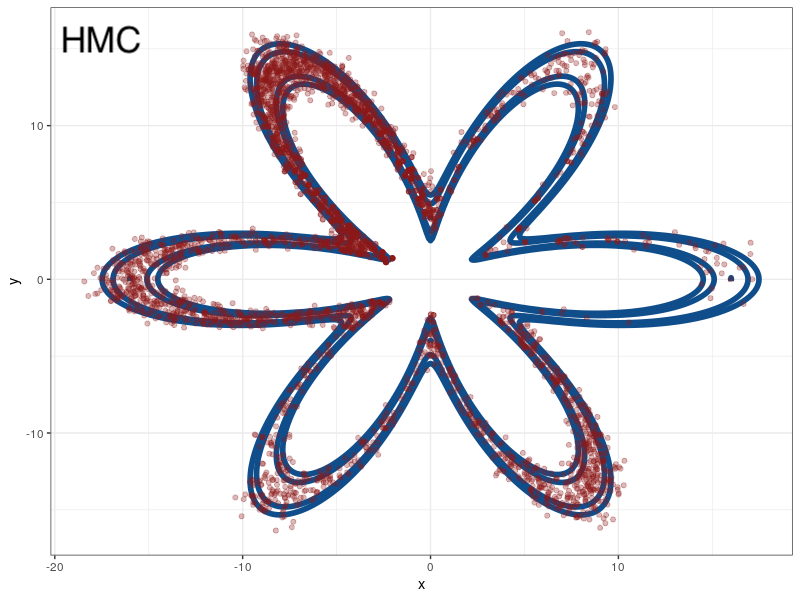}
    \includegraphics[scale=0.33]{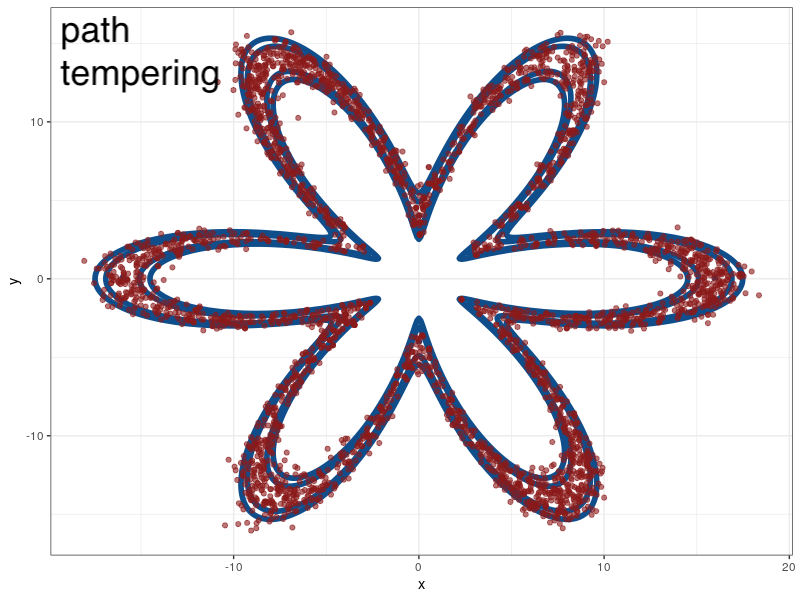}
    \caption{\em  Draws (red dots) from the flower target generated by plain HMC (left) and continuous tempering with path sampling (right). The blue contours represent the underlying target density. }
    \label{fig:flower_example_1}
\end{figure}
\begin{figure}
    \centering
    \includegraphics[scale=0.55]{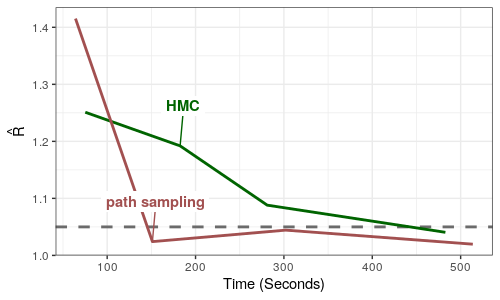}
    \includegraphics[scale=0.55]{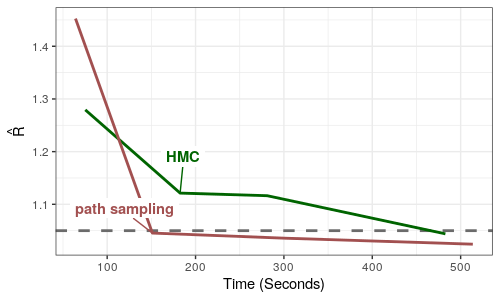}
    \caption{\em Comparison of computational cost of path sampling versus standard HMC for a flower distribution with 40 petals. The x-axis displays the time in seconds, and the y-axis shows the achieved $\widehat{R}$-value for the x coordinate (left) and y coordinate (right) of the flower distribution. The dashed grey line displays the 1.05 cutoff below which we consider our chains to have mixed sufficiently. All results are averaged over 15 replications.}
    \label{fig:flower_time_comp}
\end{figure}

We further compare the total mixing time. Figure \ref{fig:flower_time_comp} plots the computed $\hat{R}$  \citep{vehtari2019rank}  for the $x$ and $y$ coordinates of a challenging flower distribution with 40 petals against computational time. Continuous tempering with path sampling crosses the $\hat{R}=1.05$ threshold for adequate mixing in about a quarter of the time (already having included all adaptations) required for standard HMC to reach that threshold, although the latter can eventually approximate the target with many more draws. Thus, despite the much lower time cost per sample for standard HMC, path sampling still manages to be more efficient in terms of total mixing time. This disparity will only get more apparent as the difficulty of the target density increases.


\subsection{Regularized horseshoe regression: expand models continuously and efficiently}\label{sec:HS}
The proposed path sampling and continuous tempering can enhance computational efficiency even when there is no direct posterior multimodality.  In particular, the HMC mixing time scales polynomially with dimensions, hence fitting a slightly larger model can be cheaper than fitting two models separately. 
Consider a  sparse regression with  
regularized horseshoe prior \citep{piironen2016hyperprior, piironen2017sparsity}, an effective tool for Bayesian sparse regression. Letting $y_{1:n}$ be a binary outcome and $x_{n\times D}$ be predictors,
the regression coefficient $\beta$ is equipped with a  regularized horseshoe prior:
$
  \beta_d \mid \tau, \zeta, \gamma \sim  \mbox{normal}\left( 0,     {\gamma\zeta_d}{(\gamma^2 + \tau^2 \zeta_d^2)^{-1/2}}\right), 
	~ \tau \sim \mathrm{Cauchy}^+ \left(0,  {2}/ {\left(D-1)\sqrt{n}\right)}  \right), ~ \zeta_d \sim \mathrm{Cauchy}^+(0,1),  ~ d=1, \dots, D. 
$
To take into account the model freedom between the logit link $\mathrm{Pr}(y_i=1\mid \beta, \mathrm{logit})=   \mathrm{logit}^{-1}\! \left(   \sum_{d=1}^{D} \beta_d x_{id} \right)$  and  probit link $\mathrm{Pr}(y_i=1 \mid \beta, \mathrm{probit})=   \Phi\left(   \sum_{d=1}^{D} \beta_d x_{id} \right)$ in the likelihood,  we construct a tempered path between the logit and probit model  
$$p(a, \beta,\tau, \zeta, \gamma ) \propto \frac{1}{c(a)} \prod_{i=1}^n \left( \mathrm{Pr}(y=y_i|\beta, \mathrm{logit} )^{1-\lambda} \mathrm{Pr} (y=y_i|\beta, \mathrm{probit})^{\lambda}\right) p^{\mathrm{prior}}(\beta,\tau, \zeta, \gamma),~ \lambda=f(a),$$
where $p^{\mathrm{prior}}(\beta,\tau, \zeta, \gamma)$ encodes the regularized horseshoe prior.

In the first experiment, we generate $n=40$ data points and $D=100$ covariates with a  maximum pairwise correlation 0.5. Among all covariates, only the first three have nonzero coefficients $\beta_{1,2,3}$. Furthermore, $y$ is chosen to have  a logit link in the true data generating process.  We vary $a_{\min}$ in the link function $\lambda=f(a)$ as defined in Equation \eqref{eq_link},  such that 
when  $0<a<a_{\min}, \lambda=0$, the path sampling draws from the logit model, and when  $1>a>a_{\max}=1-a_{\min}, \lambda=1$, it draws from the probit model. We run path sampling with two adaptations,  $S=500$ draws in the first adaptation, and $S=3000$ in the second.  We then compute the tail effective sample size \citep*[tail-ESS,][]{vehtari2019rank} from the the probit and logit model in the aggregated draws  divided by the total sampling time (sum of 4 chains).  For compassion, we also fit these two models separately and count their effective sample size. To reflect the  bottleneck of the computation, we monitor the smallest tail effective sample size among all regression coefficients $\beta_d$. 
Path sampling between two models expands the model continuously such that both individual models are special cases of the augmented model. Because the posterior distributions $\beta|y$ under the probit and logit links are similar,  a connected path between 
them stabilizes the tail of posterior sampling, resulting in a larger unit time tail effective sample size, as shown in the left two panels of Figure \ref{fig_Horeshoe}.

In the second experiment, we fix $a_{\min}=0.35, a_{\max}= 0.65$, and vary the input dimension $D=n\rho, 2\leq \rho\leq 8$.  Given that $a_{\max}- a_{\min}=30$\% of the sampling time is spent on intermediate models, it is remarkable that most times the joint sample from path sampling renders a tail effective sample size from individual models no slower than fitting them separately, as verified in the right two panels in Figure \ref{fig_Horeshoe}. The simulation results are averaged over 10 repeated runs. 
As a  caveat, a joint path of two arbitrary  models is not always more efficient than separate fits. Figure \ref{fig_Horeshoe}  displays  the minimal  effective sample size among all regression coefficients $\beta_d$.  When averaged over all $\beta_d$, the mean tail effective sample size in path sampling is smaller  than individual fits in this example, which  can be viewed as a parameter-wise efficiency-robustness trade-off. 
 \begin{figure}
 \centering
    \includegraphics[height=1.63in]{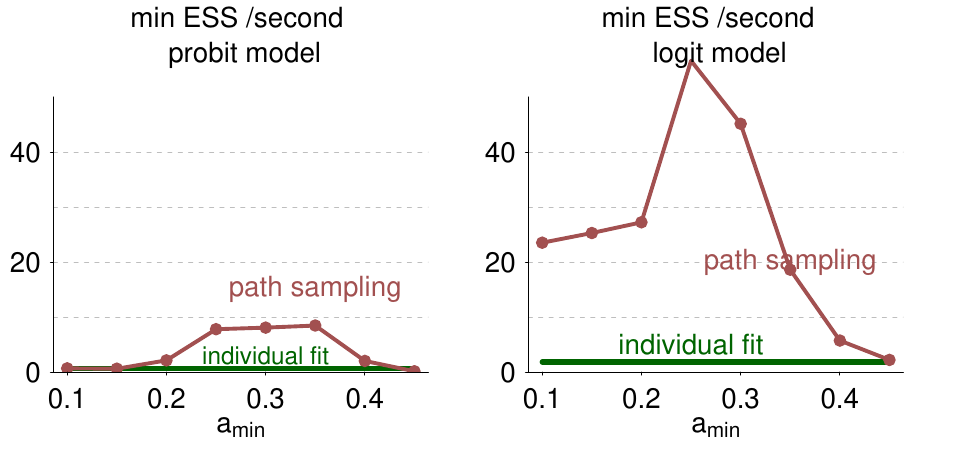} 
        \includegraphics[height=1.63in]{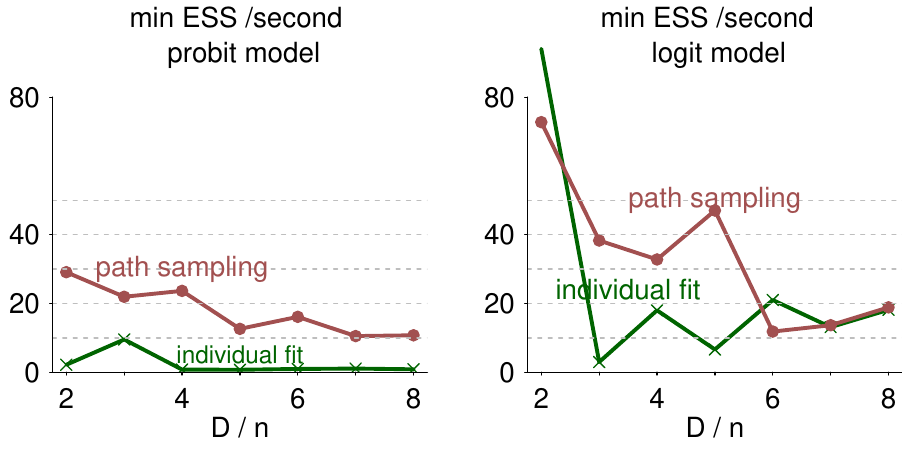} 
    \caption{\em The smallest unit time tail effective sample size among all regression coefficients in the horseshoe regression with both Bernoulli logit and probit likelihood. The green line refers to two separate model fits and the  red line corresponds to a joint tempered path sampling counting all adaptation time. Left two panels:
    we vary $a_{\min}$: the proportion of logit and probit sampling in the final joint sample.  Right two panels: we fix  $a_{\min}=0.35, a_{\max}=0.65$ and vary the input-dimension/sample-size from 2 to 8. In most cases, the joint path sampler generates a larger tail ESS /s for both probit and logit sample, in addition to all intermediate models fitted simultaneously.}\label{fig_Horeshoe}
\end{figure}

\subsection{Sampling from funnel shaped posteriors by implicit divide-and-conquer}\label{sec_exp_idc}
We apply the implicit divide-and-conquer algorithm to the hierarchical model and eight-school dataset \citep{gelman2013bayesian}. In this problem, we are given eight observed means $y_i$ and standard deviations $s_i$ which are related to each other through the following hierarchical model with unknown parameters $\theta, \mu, \tau$,
\begin{equation*}
\mathrm{centered~parameterization}:\qquad    y_i\sim \mathrm{normal}(\theta_i,s_i), \quad \theta_i \stackrel{iid}{\sim} \mathrm{normal}(\mu,\tau).
\end{equation*}
The  centered parametrization in hierarchical models often behaves as implicit left-truncation due to entropic barriers.   In particular, as $\tau\to 0$, the mass of the conditional posterior $p(\theta\mid\tau,y,s)$ concentrates around $\mu$, creating a funnel-shape in the posterior of $\theta$ that is challenging to traverse. 
In this dataset, the left truncation can be overcome by switching to a non-centered parametrization that introduces auxiliary variables $\tilde{\theta}$,
\begin{equation*}\label{eq:noncent_param}
\mathrm{non\!-\!centered~parameterization}: \qquad \theta_i = \mu+\tau\times\tilde{\theta}_i, \quad \tilde{\theta}\stackrel{iid}{\sim}\mathrm{normal}(0,1).
\end{equation*}
However, reparametrization is often restrited to location scale families, and 
choosing between centered and non-centered parametrization remains difficult for real data \citep{gorinova2019automatic}.

 \begin{figure}
     \centering
     \includegraphics[scale=0.47]{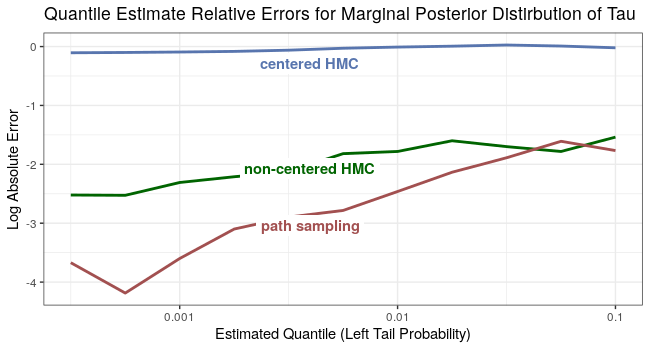}
     \includegraphics[scale=0.47]{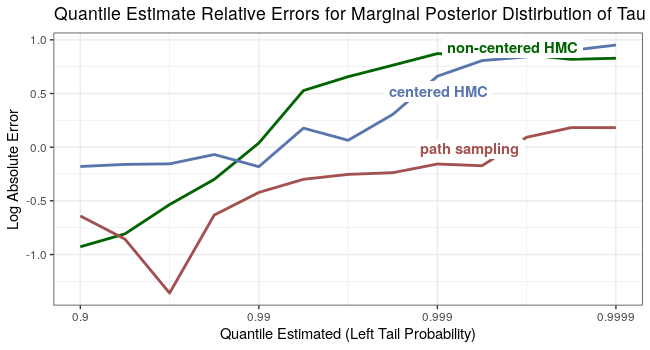}
     \caption{\em Comparison of log estimation error of the left and right tail quantile estimations using $S=4000$ posterior draws from HMC (centered and non-centered parametrizations) and implicit divide-and-conquer. The x-axis displays the left tail probability. The y-axis displays the log absolute error of the quantile estimate. The proposed method (red curve) achieves the highest accuracy.}
     \label{fig:IDC_right_tail}
 \end{figure}

   \begin{wrapfigure}[22]{r}{8cm}
    \includegraphics[width=\linewidth]{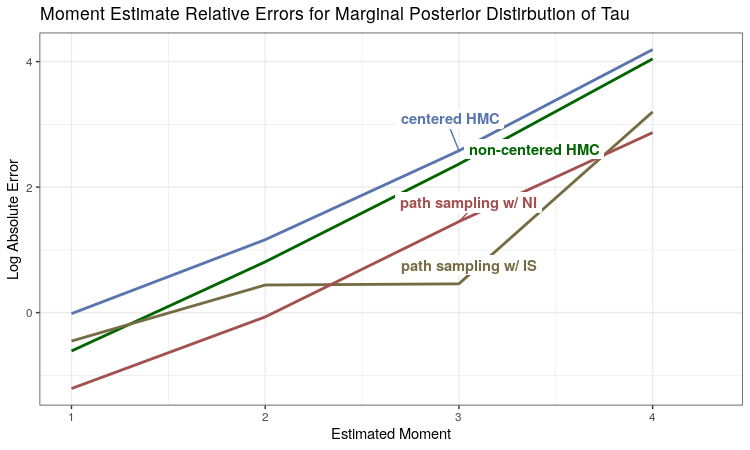}
    \caption{\em Comparison of log  estimation errors of the first to fourth posterior moment of $\tau$ using draws $S=4000$ from HMC (centered and non-centered parametrizations) or implicit divide-and-conquer. The latter returns the marginal density estimation, which is further equipped with  either importance sampling or numerical integration. The x-axis displays the estimated moment, and the y-axis displays the log  absolute error of the moment estimate.}
    \label{fig:post_moments}
\end{wrapfigure}

We apply the proposed implicit divide-and-conquer algorithm on the sample drawn from the original centered parametrization. It successively pushes the marginal of $\tau$ towards a desired target marginal $p^\mathrm{targ}(\tau)$. In our case, we use the efficiency-optimal prior of the form
$
 p(\tau) \propto \sqrt{\mathbb{E}_{\tau}U^2(\theta,\mu,\tau)},
$
where the $U(\theta,\mu,\tau)$ functions are the gradients of the log posterior given in \eqref{eq:U_gradient}. This target cannot be computed explicitly, so we also estimate it at each adaptation using a simple window-based averaging scheme over the $\tau$ draws. We adaptively repeat the algorithm until we meet the convergence criterion. 

Figure \ref{fig:IDC_right_tail} displays the log errors for quantile estimates in the right and left tails of the marginal posterior of $\tau$ for HMC with (a) centered and (b) non-centered parametrizations,  and  (c) the implicit divide-and-conquer. We set posterior draw size $S=4000$ (after thinning) in all three samples.
In these experiments, we use estimates from the non-centered parametrization with $10^6$ draws as our ground truth. In the right tail, the implicit divide-and-conquer runs out-perform the Monte Carlo estimates across all but the smallest quantile.
For the left tail, we estimate quantiles with left tail probabilities between $0.001$ and $0.1$. In this case, path sampling dominates all other methods, with a significant improvement in the extreme quantiles.


The estimated marginal density $p(\tau)$ enables expectation computation. 
We examine the performance for marginal moment estimation $\E(\tau^m)$
using both importance sampling and numerical integration (the same trapezoidal rule as in \eqref{path1}).
Figure \ref{fig:post_moments} displays the log errors of estimates for the first four moments of the marginal posterior distribution of $\tau$.  The  moments computed from simulation draws of the implicit divide-and-conquer, using either numerical integration scheme or importance sampling,  have a lower error than Monte Carlo estimates from both HMC parametrizations. 


Overall, we see from Figures \ref{fig:IDC_right_tail} and \ref{fig:post_moments} that the proposed implicit divide-and-conquer provides good performance for a range of posterior estimation tasks, often out-performing all other approaches and never generating estimates with errors large enough to be unusable in any of the tested problems. This is remarkable since the implicit divide-and-conquer scheme generates its underlying HMC samples using the centered parametrization during each adaptation.

The implicit divide-and-conquer scheme demands more  computation time than one-run HMC, but it avoids the need for a problem-specific reparametrization, so it can be applied in cases where other approaches may not be available. Furthermore, because it comes with a stopping criterion, we can ensure that the extra computation time  arrives at a sufficiently accurate result by termination.

  \section{Discussion}
 \subsection{Initialization and base measurement}
In continuous simulated tempering, we initialize the pseudo prior at $c=1$. When the underling normalizing constant $z(\lambda)$ is several orders of magnitude smaller than $z(0)$ for all  $\lambda> 0$, the sampler starting from this non-informative initialization will get stuck in $\lambda=0$. Because $f'(a)=0$ in the base, path sampling will fail to update.  In this situation, we update the slope $b_0$ in \eqref{eq_parametric_form} to be the importance sampling estimate \eqref{eq_importance_sampling}: $b_0 \xleftarrow[]{} \log \hat z^{\mathrm{IS}}(1)= \log \left( \frac{1}{S} \sum_{s=1}^S  q(\theta_s,  \lambda=1)q^{-1}(\theta_s, \lambda=0) \right)$. This estimate $\hat z^{\mathrm{IS}}$ is unlikely to be accurate, thus we only use for initialization to avoid local convergence.

In Section \ref{sec_tempering}, we merely require the base measure $\psi$ to have the same support as the target $q$.  A natural candidate for $\psi$ is the prior as we have used throughout the experiments.
Ideally, the base measurement should balance between being easy to sample from, and close enough to the target distribution. This is not a unique challenge in our method, as all (discrete) simulated tempering and, more generally, importance sampling methods need to construct a good base measurement. The current paper treats the base measurement as an extra input that user has to specify.  Based on the discussion on how to construct and optimize the proposal distribution in adaptive importance sampling \citep{geweke1989bayesian, owen2000safe, bugallo2017adaptive, paananen2019pushing}, it may be possible to obtain better performance by optimizing the choice of base measurement, which itself is often an iterative problem. 

That said, we are unlikely to be able to automatically construct a good base density for importance sampling in high dimensional problems such that the KL divergence between the base and the target is bounded. Otherwise other advanced sampling method would not be required. Fortunately, the adaptive path sampling estimate is less sensitive to base-target discrepancy. As we have seen in experiments, our path sampling-based method still yields accurate log normalizing constant  estimates and smooth joint sampling even starting from a crudely chosen base distribution when importance sampling and bridge sampling have failed.

 \subsection{Dimension limitations}
In accordance with the error analysis in Section \ref{sec_review} and the simulation results, the proposed adaptive path sampling is more scalable to higher dimensions compared with existing counterparts.  However, path sampling-based continuous tempering  can still fail in high dimensional posteriors. 

In essence, simulated tempering depends on, but can be more difficult than normalizing constant estimation. On one hand, simulated tempering does not require a  precise estimation of $z(\lambda)$ \citep*[left half of Figure \ref{fig:logz_vs_tempering}, see also][]{geyer1995annealing}, as long as there is enough posterior marginal density $p(a)$ or $p(\lambda)$ everywhere---but this will only happen when the scale of $z(\lambda)$ is small.

In Appendix C, we provide a failure mode example in a latent Dirichlet allocation model, where  the log normalizing constant estimation has a scale $\sim 10^{4}$, and path sampling estimation manages to estimate it with pointwise errors $\sim 10^{2}$. For fitting a curve, a 1\% error is accurate enough. But a pointwise 1\% error in the log normalizing constant amounts to inflating the marginal density $p(a)$ in the joint sampling \eqref{eq_joint_sample} by a factor $\exp(10^2)$ at that point, which effectively becomes a point mass and makes path tempering get stuck in one region.  Such failures happen in discrete tempering too, and is identifiable by our $\hat{k}$ diagnostics.
In other words, successful simulated tempering requires the estimation of the pointwise normalizing constant with multiplicative precision, see Figure \ref{fig:logz_vs_tempering} for an illustration.  

\begin{figure}
    \centering
    \includegraphics[width=\linewidth]{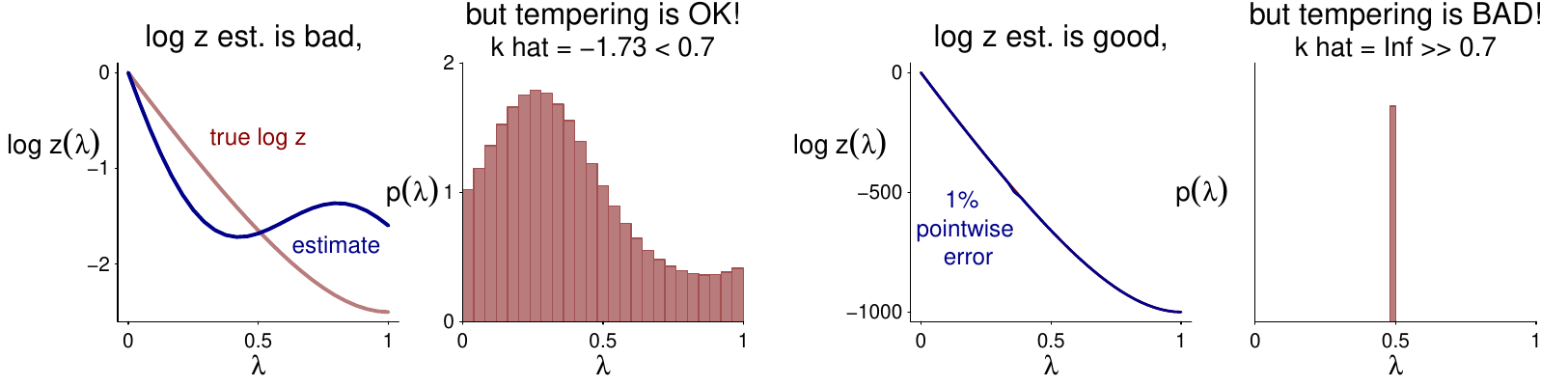}
    \caption{\em An illustration of the different orientations in normalizing constant estimation and simulated tempering. In the left two panels, the blue curve poorly fits the true log normalizing constant, but there is enough mass everywhere in the  resulting marginal density to ensure a complete path in simulated tempering. The right two panels: with a much larger scale, when log z estimation is off by 1\% at a single point, the resulting path becomes a point mass in $\lambda$. In continuous tempering, such failure is diagnosed by a large $\hat{k}$.}
    \label{fig:logz_vs_tempering}
\end{figure}

The absolute scale of the log normalizing constant is comparable to the log KL divergence between the base and the target.  In a prior-posterior tempering path, the log normalizing constant at $\lambda=1$ is the log marginal likelihood. It grows linearly with both the sample size and how closely the model fits the data (the log likelihood). When the model is poor in predication (as in the latent Dirichlet allocation example), the log normalizing constant soon escapes from the estimation accuracy that any estimate can achieve, an analogy of the ``folk theorem of statistical computing'': When you have computational problems, often there’s a problem with your model \citep{gelman2008folk}.

Furthermore,  we present  a geometric path for continuous tempering. 
It is not clear if to modify the free energy by a density-power-transformation is the best way to remove metastability, although most existing tempering methods adopt this form.
If there is rapid phase transition at some critical temperature, any power-transformation  tempering will  be prohibitively slow \citep{bhatnagar2004torpid}. Fortunately,  the general framework in Section \ref{sec_method_general} permits an arbitrary path formulations, and is easily implemented in Stan by replacing the closed form gradient $U$= log likelihood by automatic differentiation. 
We leave this more flexible tempering path for future research.   

\subsection{General recommendations}
We have developed a method that integrates adaptive path sampling and tempering and have applied it to several examples. The procedure is highly automated, and we have implemented it in an R package using the general-purpose Bayesian inference engine Stan, returning both the desired posterior  sample and the estimated log normalizing constant at convergence. See Appendix B for software details. 
 
In Bayesian computation, the ultimate goal is not to stop at the posterior simulation draws, but to use them to check and improve the model in a workflow.  If data come from an identifiable  model, then with reasonable sample size we can expect to distinguish among parameters and obtain a well behaved posterior distribution (eventually achieving asymptotic normality).  From this perspective, multimodal posteriors should be unlikely with large sample size and may represent data that do not fit the model. Hence, it is crucial to check the model fit even after the target posterior is obtained from our proposed sampling algorithms. 

Finally, although we have argued its relative advantage over existing methods, we do not think the proposed method based on adaptive sampling and tempering can solve all metasable sampling problems, either because of a badly-chosen base measurement, or the dimension and sample size limitation imposed by tempering itself. In this case, the Pareto-$\hat k$ diagnostic is still useful to understand why the method fails. Besides refining the base measurement and potentially modifying the model, another alternative strategy to metastable sampling is to use  cross validation and multi-chain stacking \citep{yao2020stacking}  to combine the non-mixed simulation draws, which in effect changes the target distribution.

 \section*{Acknowledgements}
 We thank Ben Bales, Rok Češnovar,  M\aa ns Magnusson for helpful comments and the U.S. National Science Foundation, Institute of Education Sciences, Office of Naval Research,  Sloan Foundation, and Schmidt Futures for partial support.

Replication \texttt{R} and \texttt{Stan} code for experiments is available at  
\url{https://github.com/yao-yl/path-tempering}.

Corresponds to Y.Y.\texttt{<yy2619@columbia.edu>}.

\bibliographystyle{apalike}
\bibliography{tempering}

	\normalfont
	\normalsize

\appendix
\renewcommand\thesection{\Alph{section}}
	\setcounter{section}{1}
\section*{Appendix A.  Derivation of equations}
\subsection{Thermodynamic integration identity}\label{sec_appendix_Thermodynamic}
Let $q: \R^{d+1} \rightarrow \R$ be an unnomralized density function.
Let the normalizing constant (function) $z:\R \rightarrow \R$ be defined by
$
z(\lambda) = \int_{\R^d} q(\lambda, \theta) d\theta,
$
where $\theta\in\R^d$. Calculus shows 
\begin{equation}\label{eq_thermodynamic}
\begin{split}
\frac{d}{d\lambda} \log z(\lambda) 
&= \frac{\frac{d}{d\lambda} \displaystyle \int_{\R^d} q(\lambda, \theta) d\theta}
{\displaystyle \int_{\R^d} q(\lambda, \theta) d\theta} \\
&= \frac{\displaystyle \int_{\R^d} \frac{\partial}{\partial \lambda}  q(\lambda, \theta) d\theta}
{z(\lambda)} \\
&= \frac{\displaystyle \int_{\R^d} 
\frac{\partial}{\partial \lambda} \log q(\lambda, \theta) q(\lambda,\theta) d\theta}
{z(\lambda)} \\
&= \displaystyle \int_{\R^d} 
\frac{\partial}{\partial \lambda} \log q(\lambda, \theta) \frac{q(\lambda,\theta)}
{z(\lambda)}  d\theta, \end{split}
\end{equation}
That is, 
$\frac{d}{d\lambda} \log z(\lambda) = \E_{\theta|\lambda}(\frac{\partial}{\partial \lambda}
\log q(\lambda, \theta))$, which  leads to the thermodynamic integration identity \eqref{eq_score} we use by chain rule.

At the second equation of \eqref{eq_thermodynamic}, we are assuming the legitimacy of changing the order of derivative and integral. One sufficient condition is that there exists a sufficiently large constant $M$ such that  $\int_{\R^d} |\frac{\partial}{\partial \lambda}  q(\lambda, \theta) | d\theta < M$ for all $\lambda$, and the interchangeability follows by the dominated convergence theorem.

\subsection{The link function in continuous tempering}\label{sec_appendix_link}
In continuous tempering,  we choose the following piecewise polynomial link function (see Figure \ref{fig_link} for a visualization):
\begin{equation}\label{eq_link}
f(a)= \begin{cases}
0,     &0\leq a <  a_{\min}, \\
 -2 (\frac{a-a_{\min}}{a_{\max}-a_{\min}})^3 + 3 (\frac{a-a_{\min}}{a_{\max}-a_{\min}})^2,    &a_{\min} \leq a < a_{\max}, \\
1,    &a_{\max} \leq a < 2-a_{\max},\\
 -2 (\frac{2-a_{\min}-a}{a_{\max}-a_{\min}})^3 + 3 (\frac{2-a_{\min}-a}{a_{\max}-a_{\min}})^2,    &2-a_{\max} \leq a <2-a_{\min}, \\
0,    & 2-a_{\min}\leq a  \leq 2.
\end{cases}
\end{equation}
It has a continuous first order derivative.  In experiments and default software implementation, we set $a_{\min}=0.1$ and $a_{\max}=0.8$.

  \subsection{Rao-Blackwellization as the optimal multistate bridge sampling}\label{sec_rb_bri}

In this subsection we follow the notation of bridge sampling discussed in \citet{Lelievre2010Free}.  Let $q_i= q(\theta|\lambda_i)$ be the unnormalized density at a sequence of discrete temperatures $\lambda_i$, $i=1, \dots I$, and   
  $Z= (z_2,\dots, z_I)^T$   the (ratio of) normalizing constants (assuming $z_1=1$).   With an arbitrary list of $\Theta\to \R$ functions $\alpha_{i,j}$,  we define  $A$ an $(I-1)\times(I-1)$ matrix,  and B an $(I-1)$ vector:
  $$A= \begin{bmatrix}
  a_{2}       & -b_{23} & \dots & -b_{2I} \\
  -b_{32}       & a_3 &  \dots & -b_{3I} \\
  \hdotsfor{4} \\
  -b_{I2}       & -b_{I3} &  \dots & a_I
  \end{bmatrix}, 
  \quad
  B= \begin{bmatrix}
  b_{21}\\
  b_{31}  \\
  \vdots \\
  b_{I1}
  \end{bmatrix},
  $$
  with each entry $a,b$ a shorthand for 
  $$b_{ij} = \E_{\pi_j} (\alpha_{ij} q_i),  \quad a_i=\sum_{j=1, j\neq i}^I\E_{\pi_i}(\alpha_{ij} q_j).$$
In discrete tempering, each time we sample from $q(\theta, \lambda)=\frac{1}{c(\lambda)}q(\theta, \lambda)$. Since $\lambda$ is discrete here, we denote $c_m= c(\lambda_{m})$ and  $q_m= q(\theta, \lambda= \lambda_{m})$,  both of which are given in each adaptation.  \citet{Carlson2016Partition} considered the  Rao-Blackwellized estimate of the normalizing constant:
$$\hat z^{RB}_{k}  =    \sum_{l=1}^n \frac{q_k (\theta_l)}{ \sum_{j=1}^I  c_j q_j(\theta_l) },\quad  k=1, \dots, I. $$
\begin{proposition}
The Rao-Blackwellized estimate 
can be derived from  multistate bridge sampling by choosing
$$ \alpha_{ij} (\theta)=  \frac{n_j  \hat z_j^{-1} }{ \sum_m  c_m^{-1} q_m (\theta)  }, $$
which is further an empirical estimate of the optimal bridge sampling functions. 
\end{proposition}

\begin{proof}
 We rearranged the  multistate bridge sampling estimates
  $z_i b_{ji}= z_j b_{ij}, ~ \forall i,j$
into a matrix form 
  $$AZ=B.$$ 
 
\citet{shirts2008statistically} showed that the optimal sequence  of functions that minimizes the variance of the estimated $Z$ is 
$$\alpha_{ij} (\theta)=  \frac{n_j  z_j^{-1} }{ \sum_m  n_m z_m^{-1} q_m (\theta)  }, $$ 
which in practice, starting from some initial guess $\hat z$,  this can be approximated by 
   $\alpha_{ij} (\theta)=  \frac{n_j  \hat z_j^{-1} }{ \sum_m  n_m \hat z_m^{-1} q_m (\theta)  }. $

  Denote $S(\theta)= { \sum_m  c_m q_m (\theta)  }$, then   $ \alpha_{ij} (\theta)=\frac{n_j  \hat z_j^{-1}} {S(\theta)}$. We estimate the matrices $A$ and $B$ by their empirical means.
  \begin{align*}
  \hat a_i \hat z_i &=  n_i^{-1}  \sum_{k=1}^{n_i} \hat z_i  \sum_{j=1, j\neq i}^I   \alpha_{ij}(\theta_{i, k} )q_j (\theta_{i,k}) \\
  &=\hat z_i n_i^{-1}    \sum_{k=1}^{n_i} \frac{  \sum_{j=1, j\neq i}^I  n_j  \hat z_j^{-1} q_j (\theta)  }{S(\theta_{i,k})}\\
  &= \hat z_i -  \sum_{k=1}^{n_i}  \frac{q_i (\theta_{i,k})}{S(\theta_{i,k})}.
  \end{align*}
  \begin{align*}
  \sum_{j=1, j\neq i}^{I} \hat b_{ij} \hat z_{j}&=  \sum_{j=1, j\neq i}^{I}  \hat z_{j} n_i^{-j}\sum_{k=1}^{n_j} \frac{n_j  \hat z_j^{-1} q_i (\theta_{j,k} )  } {S(\theta_{j,k})}
  =\sum_{j=1, j\neq i}^{I}   \sum_{k=1}^{n_i}  \frac{ q_i (\theta_{j,k} )  } {S(\theta_{j,k})}.
  \end{align*}
  
  Combining these two parts, we obtain the final estimate
  \begin{align*}
  z_i &= \sum_{j=1, j\neq i}^{I}   \sum_{k=1}^{n_i}  \frac{ q_i (\theta_{j,k} )  } {S(\theta_{j,k})} +  \sum_{k=1}^{n_i}  \frac{q_i (\theta_{i,k})}{S(\theta_{i,k})}\\
  &=  \sum_{m=1}^N  \frac{q_i (\theta_{m})}{   \sum_{j=1}^I c_j  q_j (\theta_m)   },
  \end{align*}
  which is identical to the Rao-Blackwellized estimates.
  \end{proof}
  
 \subsection{Path sampling as the  limit of bridge sampling and annealed importance sampling}\label{sec_limit}
\begin{proposition}
Path sampling can be viewed as the continuous limit of  bridge sampling \eqref{eq_bridge} and annealed importance sampling \eqref{eq_jarzynski} when the intermediate states $(0=\lambda_0 < \lambda_1 < \ldots < \lambda_{L+1}=1)$ is infinitely dense, such that $\max_l \delta_l \to 0$, where $\delta_l = \lambda_{l+1}- \lambda_{l}$ is the neighboring spacing. 
\end{proposition}

\begin{proof}
The proof is similar to the reasoning in \citet{Gelman1998Simulating}.
Notably, bridge sampling and annealed importance sampling  work in the scale of the normalizing constant $z$, essentially computing $z(1)/z(0)$  by    $ \prod_{l=0}^L \left( z(l+1)/z(l)\right)$, or equivalently, 
$\log z(1)/z(0) = \sum_{l=0}^L \log z(l)- \log z(l-1)$.  Further, both bridge sampling and annealed importance sampling are based on importance sampling identity (with potential refinement of more intermediate states in  bridge sampling):
$$ \frac{z(l)}{z(l-1)} = \int_{\Theta} \frac{q(\theta|\lambda_{l+1})}{q(\theta|\lambda_l)}p(\theta|\lambda_l)d\theta.$$
In general, $\log  \E\left({q(\theta|\lambda_{l+1})}/{q(\theta|\lambda_l)}\right) \neq   \E\left( \log {q(\theta|\lambda_{l+1})} - \log {q(\theta|\lambda_l)}\right)$, where the expectation is taken over  $\theta \sim p(\theta|\lambda_l)$. However, such difference will be be approaching  zero when we have fine ladder. 

For a fixed $\lambda_{l}$, let $$G_l(\xi)= \log \int_{\Theta} \frac{q(\theta|\lambda_{l}+ \xi) }{q(\theta|\lambda_l)}p(\theta|\lambda_l)d\theta.$$  It satisfies $G_l(0)= 0$, and its derivative is 
  \begin{align*}
G_l'(\xi) &= \frac{d}{d\xi} \left(\log \int_{\Theta} \frac{q(\theta|\lambda_{l}+ \xi) }{q(\theta|\lambda_l)}p(\theta|\lambda_l)d\theta \right)\\
&= \frac{z(l)}{z(l+\xi)}  \left(\int_{\Theta} \frac{\partial }{\partial \xi} \frac{q(\theta|\lambda_{l}+ \xi) }{q(\theta|\lambda_l)}p(\theta|\lambda_l)d\theta \right).
\end{align*}
At $\xi=0$, $G_l'(0)$  becomes identical to the path sampling gradient in  \eqref{eq_gradient}, as
$$G_l'(0)=  \int_{\Theta} \frac{\partial}{\partial \lambda}  \log q(\theta|\lambda_{l}) p(\theta|\lambda_l)d\theta.$$ 

By Taylor series expansion, 
$G_l(\xi)= G_l(0) +  \xi \int_{\Theta} \frac{\partial}{\partial \lambda} \log q(\theta|\lambda_{l}) p(\theta|\lambda_l)d\theta + o(\xi),$ Hence, in the limit as $\max_l \delta_l \to 0$, the importance sampling based estimate can be rearranged into  
\begin{align*}
\log z(1)/z(0) &= \sum_{l=0}^L G_l(\delta_l) \\
&=\sum_{l=0}^L   \left( \delta_l \int_{\Theta} \frac{\partial}{\partial \lambda} \log q(\theta|\lambda_{l}) p(\theta|\lambda_l)d\theta + o(\delta_l) \right)\\
&= \int_{0}^1 \int_{\Theta} \frac{\partial}{\partial \lambda} \log q(\theta|\lambda_{l}) p(\theta|\lambda_l)d\theta d \lambda + o(1),
\end{align*}
where the dominant term equals the path sampling estimate, and the remainder approaches 0 in the dense limit $\delta_l \to 0, \forall l$ since $\sum_l {\delta_l}=1$. 
\end{proof}

 \subsection{On the choice of prior  \texorpdfstring{$p(a)$}{}}\label{sec_prior}
In path sampling, the final marginal distribution of $a$ relies on user specification. By default we use $z(\cdot) \to c(\cdot)$, which enforce a uniform   marginal distribution.  More generally, by updating
$c(\cdot) \leftarrow z(\cdot)/p^{\mathrm{prior}}(\cdot)$, the final marginal distribution of $a$ will approach $p^{\mathrm{prior}}$. In this section we discuss the choice of  $p^{\mathrm{prior}}$ beyond uniformity.
 
The choice of $p^{\mathrm{prior}}$   is subject to a efficiency-robustness trade-off.  There are three separate goals to pursue via prior tuning:

\begin{enumerate}
    \item \textbf{Robustness.} Because a uniform $a$ ensures the Markov chain has explored the whole temperature space,  it is a conservative choice and we use for default in adaptations. $$
 p(a) \propto 1.
$$
\item \textbf{Minimal variance of log $z$.}
On then other end of the spectrum, we can ask for efficiency.   \cite{Gelman1998Simulating} proved that the generalized Jeffreys prior  minimizes the variance of estimated log normalizing constant.
$$
 p^{\mathrm{opt}}(\lambda) \propto \sqrt{\E_{\theta|\lambda}U^2(\theta,\lambda)}.
$$
where $U(\theta,\lambda)= \frac{\partial}{\partial \lambda} \log q(\theta, \lambda)$. 

\item \textbf{Smooth transition in the joint sampling.}
From the perspective of successful sampling, we could ask for   
 $$\mathrm{KL}\Bigl( \pi_a, \pi_{a+\delta a}  \Bigr)\approx \mathrm{constant},$$
which is related to the constant acceptance rate in discrete Gibbs update (when the discrete temperature update is restricted to either a one-step jump $\lambda_k\to \lambda_{k\pm 1}$ or remain unchanged, This constant acceptance rate  is often used as a tuning target in discrete tempering \citep{geyer1995annealing}. The next proposition gives an closed form optimal prior that ensures this constant KL gap.
\end{enumerate}

\begin{proposition}
 The desired prior to achieve a smooth KL gap is
$$
p^{\mathrm{opt}}(a)   \propto \frac{1}{{f'(a)}} \sqrt{ \mathrm{Var}_{\theta\sim p(\theta|a)}  U(\theta, a) }, ~\forall a_{\min}< a < a_{\max}.
$$
This is similar to the minimal-variance prior above, with a slight twist from the second moment to variance. 
\end{proposition}

\begin{proof} 
It is easy to verify that 
$$\mathrm{KL}\Bigl( \pi_a, \pi_{a+\delta a}  \Bigr) = \int \log \bigl( \frac{\pi_{a}}{\pi_{a+\delta a}}   )d\pi_{a} = \frac{1}{2}(\delta a)^2  \Bigl( \frac{d^2}{d a^2} \log z(a) - \frac{f''(a)}{f'(a)}   \frac{d}{da}\log(z(a)) \Bigr) + o(\delta a)^2.$$
Assuming we have already sampled from the joint stationary distribution,  the gap between two neighboring order statistics reflects how dense the  local density is, i.e., 
$\delta a \propto 1/ p(a).$

Further, the two derivative terms can be expressed by expectations,
$$ \frac{d}{d a} \log(z(a))= f'(a) \mathrm{E}_{\theta\sim p(\theta|a)} \bigl[  \log(\psi)-\log(\phi )   \bigr].  $$
\[\begin{split}
\frac{d^2}{d a^2} \log(z(a))=& f''(a) \mathrm{E}_{\theta\sim p(\theta|a)} \bigl[  \log(\psi)-\log(q )   \bigr]+ f'(a)^2 \mathrm{E}_{\theta\sim p(\theta|a)} \left(  \bigl( \log(\psi)-\log(q )   \bigr)^2  \right)  \\
 &-\Bigl(   f'(a) \mathrm{E}_{\theta\sim p(\theta|a)}\left( \log(\psi)-\log(q )   \right)  \Bigr)^2,
 \end{split}\]  
 which further simplifies to 
\[\begin{split}
&\frac{d^2}{d a^2} \log(z(a)) - \frac{f''(a)}{f'(a)}   \frac{d}{da}\log(z(a)) \\
= & f'(a)^2 \mathrm{E}_{\theta\sim p(\theta|a)} \left( \bigl( \log(\psi)-\log(q )   \bigr)^2  \right) -\Bigl(   f'(a) \mathrm{E}_{\theta\sim p(\theta|a)}\left( \log(\psi)-\log(q )   \right)  \Bigr)^2. 
 \end{split}\]  
Put all together, the constant KL gap will be achieved by
\[\begin{split}
p(a)   \propto& (\frac{d^2}{d a^2} \log z(a)  - \frac{f''(a)}{f'(a)}   \frac{d}{da}\log(z(a)) \Bigr  )^{\frac{1}{2}} \\
=& \frac{1}{{f'(a)}}   \Bigl( \mathrm{Var}_{\theta\sim p(\theta|a)} \bigl(\log(\psi)-\log(q )\bigr)  \Bigr)^{\frac{1}{2}}.
  \end{split}\] 
 It is also evident that under this prior, both $\mathrm{KL}\Bigl( \pi_a, \pi_{a+\delta a}  \Bigr)$ and the reserve jump  $\mathrm{KL}\Bigl(\pi_{a+\delta a}, \pi_a \Bigr)$ approximates (different) constants along the trajectory. 
 \end{proof}  
 
Due to dependence on the unknown normalizing constant (and higher orders),  these two efficiency-optimal priors require additional tuning and adaptations.  In general,  we still prefer to use the simple uniform margin for robustness. Nevertheless, our method enables any user specific-prior choice, and in Section \ref{sec_exp_idc}, we illustrate that  this adaptively estimated and assigned efficiency-optimal prior further reduces the variance in implicit divide and conquer scheme.

\subsection{Choice of regression kernels}
In regularized path sampling (Section \ref{sec_method}), we regularize  $\log z$ and $\log c$ via a  parametric  regression form: 
\begin{equation}\label{eq_kernel}
   \min_{\beta}\sum_{i=1}^{I}\left(\log z(f(a^{*}_i)) - \left(\beta_0 f(a^{*}_i) + \sum_{j=1}^J  \beta_j  \gamma_j(f(a^{*}_i) )\right)\right)^2.
 \end{equation}
In other words, we approximate $\log z(f(a))$ by $\beta_0 f(a) + \sum_{j=1}^J  \beta_j  \gamma_j(f(a)).$ This also enables a cloded form expression for the gradient $c'$ and $z'$ that will be used in estimate \eqref{eq_gradient_marginal}.

The term $\log z(f(a))$ will be a constant function wherever $f(a)$ is constant.  
In our experiment, we have tried a sequence of Gaussian kernels, 
$$\gamma_j(\lambda) = \exp ( \frac{ -(\lambda- \lambda^{{\mathrm{kernel}}}_{j} )^2 }{2\sigma^2_{\mathrm{kernel}}} ),  ~\lambda^{\mathrm{kernel}}_j = \frac{j}{J+1}, ~ \sigma_{\mathrm{kernel}}= 1/J, ~  1\leq j\leq J, $$
logit kernels,  
$$\gamma_j(\lambda) =  \frac{1}{1+ \exp (-\frac{\lambda- \lambda^{{\mathrm{kernel}}}_{j}  }{\sigma_{\mathrm{kernel}}} ) },  ~\lambda^{\mathrm{kernel}}_j = \frac{j}{J+1}, ~ \sigma_{\mathrm{kernel}}= 1/J, ~  1\leq j\leq J, $$
and cubic splines.

We have not found the kernel choice to have a large impact on the final tempering procedure or log normalizing constant. For the experiments in Section \ref{sec_exp}, we used a combination of the Gaussian and logit kernels at $J=10$ points within the path sampling adaptation steps for speed and simplicity. After running our final adaptation, we then smoothed the final estimate of the normalizing constant or marginal posterior using cubic splines since these introduce less pseudo-periodic behavior than the Gaussian kernels.

To be clear, although $z(\cdot)$ can be further used in density estimation,  the kernels are used here for regularization and functional approximation, and is \emph{not} relevant to \emph{kernel} density estimation. The latter is noisy and contingent on various smoothness assumptions.
 In contrast, \eqref{eq_kernel} is a \emph{linear regression} problem as $\log z(f(a^{*}_i))$ is known from path sampling estimate. Besides, the choice of inner kernel points should  not be confused with the tempering ladder. Here we are sampling a and $\lambda$ continuously and evaluate $z(\lambda)$ for all $\lambda\in [0,1]$.

 \section*{Appendix B.  Software implementation in \texttt{Stan}}
 To provide an \texttt{R}  \citep{R} interface of   path sampling and continuous tempering, 
we create a  package \texttt{pathtemp}, with the  underlying execution inside  the general-purpose Bayesian inference engine \texttt{Stan} \citep{stan_jss,  stan}. The source code is available at  \url{https://github.com/yao-yl/path-tempering}.
The procedure is highly automated and requires minimal tuning.  
  
To install the package, call
\begin{verbatim} 
devtools::install_github("yao-yl/path-tempering/package/pathtemp",upgrade="never")
\end{verbatim} 
  
 We demonstrate the practical implementation of continuous tempering on a  Cauchy mixture example. Consider the following Stan model:
\begin{verbatim} 
data {
  real y;
}
parameters {
  real theta;
}
model{
  y ~ cauchy(theta, 0.2);   
 -y ~ cauchy(theta, 0.2);   
}
\end{verbatim}

\citet{yao2020stacking} have analyzed the posterior behaviour of this mixture model. With a moderately large input data $y$, the posterior distribution of $\theta$ will asymptotically concentrated at two points close to $\pm y$. As a result,  Stan cannot fully sample from this two-point-spike even with a large number of iterations. 

To run continuous tempering, a user can  specify any base model,  say $\theta\sim$ normal$(0,5)$,  and list it in an  \texttt{alternative model} block as if it is a regular model.   
 
\begin{verbatim} 
...
model{     // keep the original model  
  y ~ cauchy(theta,0.2);   
 -y ~ cauchy(theta,0.2);   
}
alternative model{    // add a new block of the base measure (e.g., the prior). 
 theta ~ normal(0,5);   
}
\end{verbatim}

After saving this code to a stan file \texttt{cauchy.stan}, we run the function  \texttt{code\_temperature\_augment()}, which automatically constructs a tempered path between the orginal model and the alternative model, and generates a working model named \texttt{cauchy\_augmented.stan}:
\begin{verbatim} 
library(pathtemp)
update_model <- stan_model("solve_tempering.stan")
file_new <- code_temperature_augment("cauchy.stan")
> output:
> A new stan file has been created: cauchy_augmented.stan.
\end{verbatim}

We have automated path sampling and its adaptation into a function \texttt{path\_sample()}. The following two lines realize adaptive path sampling.

\begin{verbatim} 
sampling_model <- stan_model(file_new) #  translated to C++ code
path_sample_fit <- path_sample(data=list(gap=10), # data list in original model
                            sampling_model=sampling_model)
\end{verbatim}

The returned value \texttt{path\_sample\_fit} provides access to the posterior draws $\theta$ from the target density and base density, the join path in the $(\theta, a)$ space in the final adaptation, and the estimated log normalizing constant $\log z(\lambda)$.

\begin{verbatim} 
sim_cauchy <- extract(path_sample_fit$fit_main)
in_target <- sim_cauchy$lambda==1
in_prior  <- sim_cauchy$lambda==0
# sample from the target 
hist(sim_cauchy$theta[in_target])
# sample from the base 
hist(sim_cauchy$theta[in_prior])
# the joint "path"
plot(sim_cauchy$a, sim_cauchy$theta)
# the normalizing constant
plot(g_lambda(path_sample_fit$path_post_a), path_sample_fit$path_post_z)
\end{verbatim}
The output is presented in Figure \ref{fig:cauchy_output}.
\begin{figure}[ht]
    \centering
    \includegraphics[width=\linewidth]{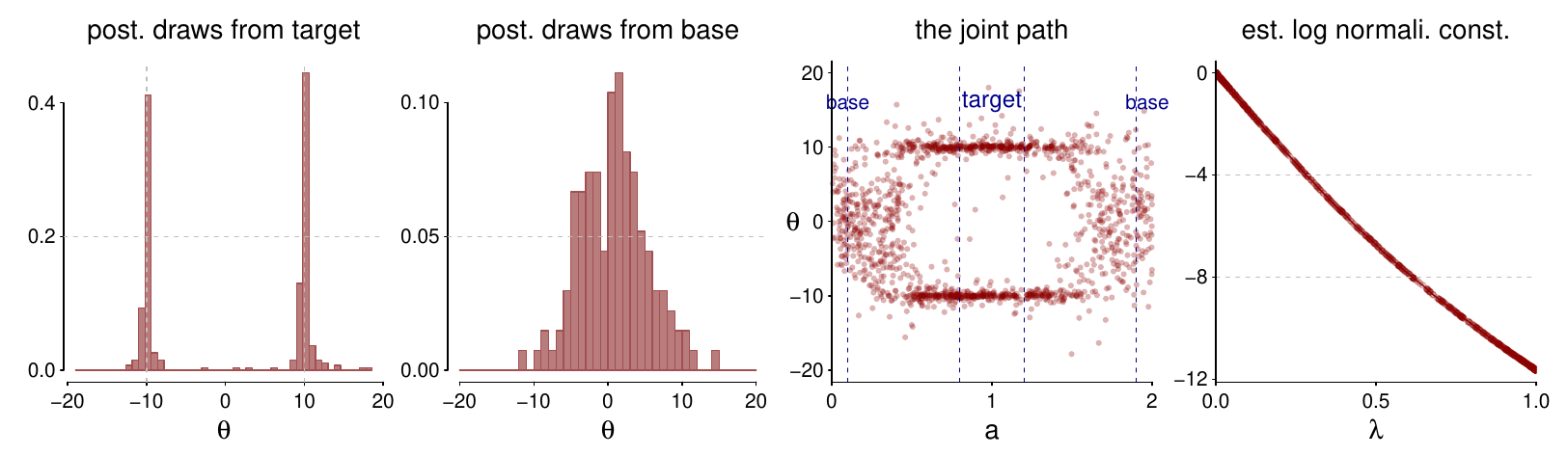}
    \caption{\em Output from the Cauchy code example: Posterior draws $\theta$ from the target density and base density, the join path in the $(\theta, a)$ space in the final adaptation, and the estimated log normalizing constant $\log z(\lambda)$. The visualization is based on 6 adaptations and S=1500 posterior draws.}
    \label{fig:cauchy_output}
\end{figure}

Second, this automated procedure enables to fit two models together. 

The following Stan code fits a regression with both the probit and logit link. A path between them effectively expands the model continuously such both individual model are special cases of the augmented model. The computational efficiency is enhanced as we are fitting one slightly larger model rather than fitting two models. In addition,  the log normalizing constant tells us which models fits the data better, which is related to but distinct from the log Bayes factor in model comparisons. 
The output in one run is presented in Figure \ref{fig:logit_output}. The data favor the logit link accordingly.
\begin{figure}[ht]
    \centering
    \includegraphics[width=\linewidth]{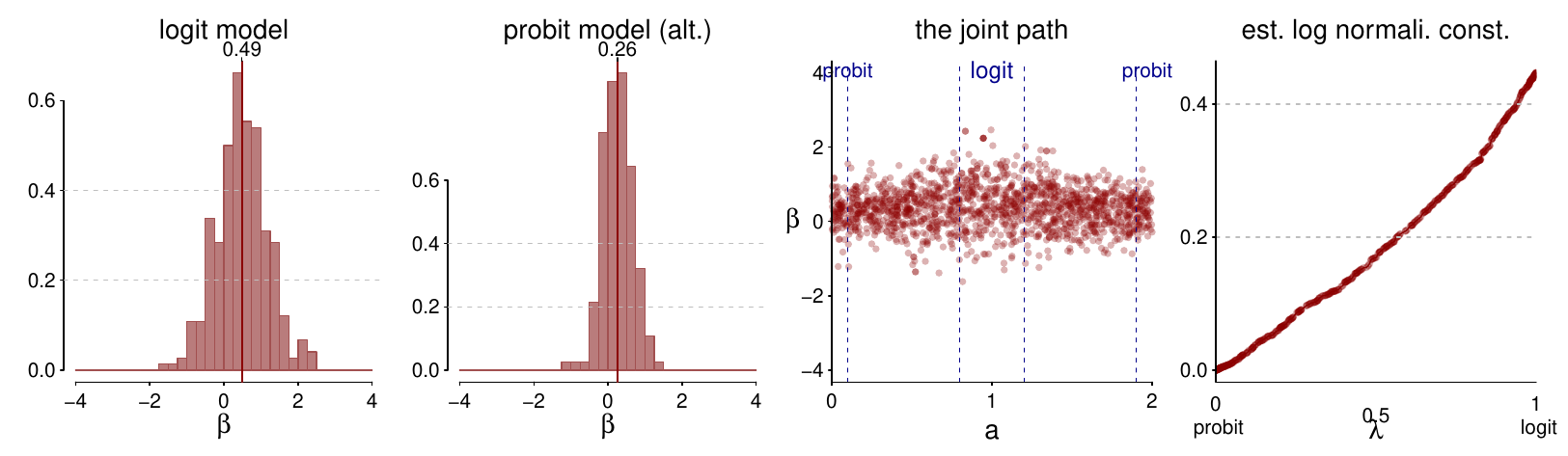}
    \caption{\em Output from the logit-probit code example: the posterior draws $\beta$ from the probit and logit link, the join path in the $(\beta, a)$ space in the final adaptation, and the estimated log normalizing constant $\log z(\lambda)$. The visualization is based on 2 adaptations and S=1500 posterior draws.}
    \label{fig:logit_output}
\end{figure}

\begin{verbatim} 
data {
	int n;
	int y[n];
	real x[n];
}
parameters {
	real beta;
}
model {
    beta ~ normal (0,2);
    y ~ bernoulli_logit(beta * x);  // logistic regression
}
alternative model {
    beta ~ normal (0,1);  // can be a different prior
    y ~  bernoulli(Phi(beta * x)); // probit regression
}
\end{verbatim}

\section*{Appendix C.  A failure mode: simulated tempering on latent Dirichlet allocation}
Here we present a failure mode of simulated tempering on a high dimensional latent Dirichlet allocation (LDA) model, which is widely used in natural language processing, computer vision, and population genetics. In the model, the $j$-th document ($1\leq j \leq J$) is drawn from the $l$-th topic ($1\leq l \leq L$) with probability  $\theta_{jl}$, where the topic is defined by a vector of probabilities $\phi_l$ over the vocabulary, such that each word in the document from topic $l$ is independently drawn from a multinomial distribution with probability $\phi_l$.  	We apply the LDA model to texts in the novel {\em Pride and Prejudice}. After removing frequent and rare words, the book contains 2025 paragraphs and 32877 words, with a total unique vocabulary size of 1495. We randomly split the words in the data into a training and  test set. The dimension of the parameters  $\theta$ and $\phi$ grows as a function of the number of topics $L$ by $2025\times L$ and $L\times 1495$ respectively.  We place independent Dirichlet$(0.1)$ priors on $\theta$ and $\phi$. In our experiment, we use $L=5$ topics.

LDA is prone to a multimodal posterior distribution.  The variational based inference often does not replicate itself from multiple runs or data shuffle, which can create the appearance of random results for the user and reduces the predictive power. \citet{yao2020stacking} has reported posterior multimodality in this dataset and model setting. 
	
We use continuous tempering to sample from the joint distribution proportional to the joint density: $c^{-1}(\lambda)$prior$^{1-\lambda}$likelihood$^{\lambda}$. For initialization, we start by simulating draws $(\theta)_{i}^{\mathrm{prior}}$ in the prior, where $\theta$ now denotes all parameters in the model, 
and assign the importance sampling estimate   $\log\left(1/S\sum_{i=1}^Sp(y|( \theta)_{i}^{\mathrm{prior}})\right)$ to the initial slope coefficient $b_0$, which is close to $-7\times10^{4}$ in our first run.

After 10 adaptations and 3000 joint HMC iterations per adaptation, the sampler still explores a thin range of the transformed temperature $a$, both in the last sample, and all 10 samples mixed, shown in the two histograms in Figure \ref{fig:lda1}. This marginal pattern fails the Pareto-$\hat k$ diagnostic, suggesting the joint tempering path is unreliable.  

\begin{figure}[ht]
    \centering
    \includegraphics[width=\linewidth]{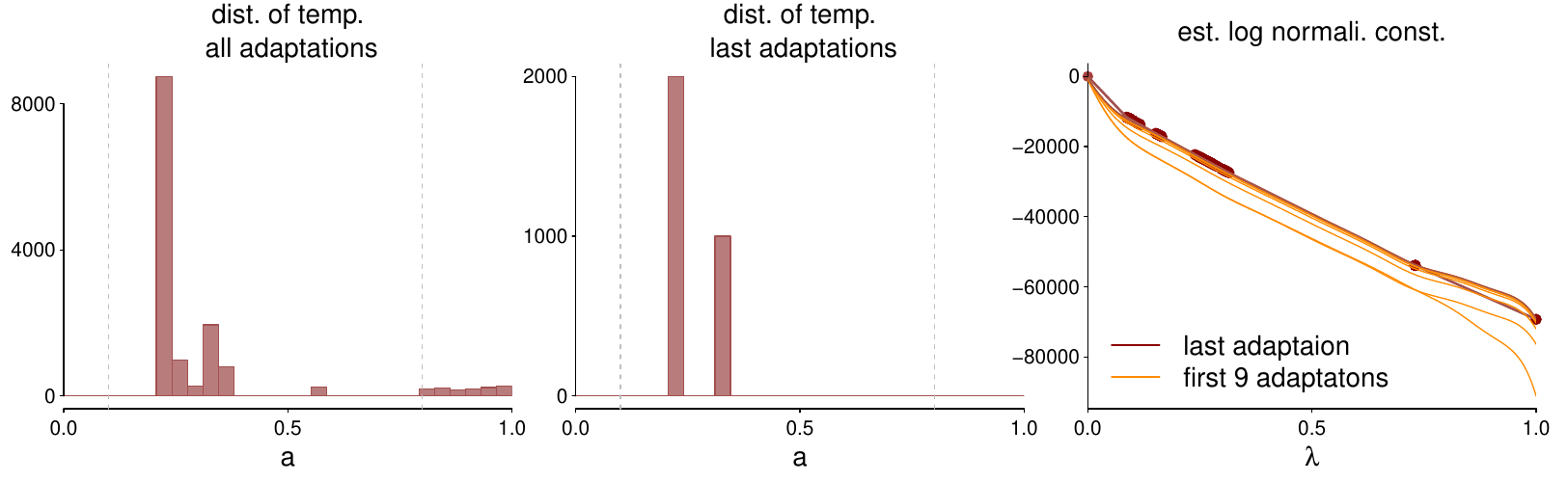}
    \caption{\em 1-2: Histograms of the transformed temperature among all adaptations and the final adaptation. 3: The estimated log normalizing constant as a function of temperature $\lambda$.}
    \label{fig:lda1}
\end{figure}

In adaptive path sampling, the log normalizing constant is  computed by the integral of the pointwise gradient $u_i= \log p(y|a_i, \theta_i) \frac{d}{da}\lambda\bigr|_{a=a_i}.$ 
We examined these three terms:  log likelihood $\log p(y|a_i, \theta_i)$, pointwise gradient $u_i$, and the $\lambda$ derivative   $\frac{d}{da}\lambda\bigr|_{a=a_i}$ along all sampled $a_i$ in Figure \ref{fig:lda2}. The variation in log likelihood $\log p(y|a_i, \theta_i)$ could be a concern as we approximate its pointwise expectation by  one Monte Carlo draw.  However, the variation in pointwise log likelihood (scale of $10^3$) is small compared with its absolute scale ($\sim 10^5$). Hence, the gradient seems to have been computed with  small Monte Carlo error (panel 2).  
Expect if we zoom in, the potinwise variation  can still be found around  $a\approx 0.6$ (panel 4). 

Why does this matter? Recall what we typically require for \emph{curve fitting} in data analysis. We often do not care about the absolute scale of the curve. Indeed we often transform all input to the unit scale in regressions, implicitly assuming that approximating a process N$(1:10, 1)$ with pointwise errors N$(0,0.1^2)$ is operationally equivalent to approximating the multiplicatively inflated process 
N$(100:1000, 100^2)$ with pointwise errors N$(0,10^2)$.

\begin{figure}[ht]
    \centering
    \includegraphics[width=\linewidth]{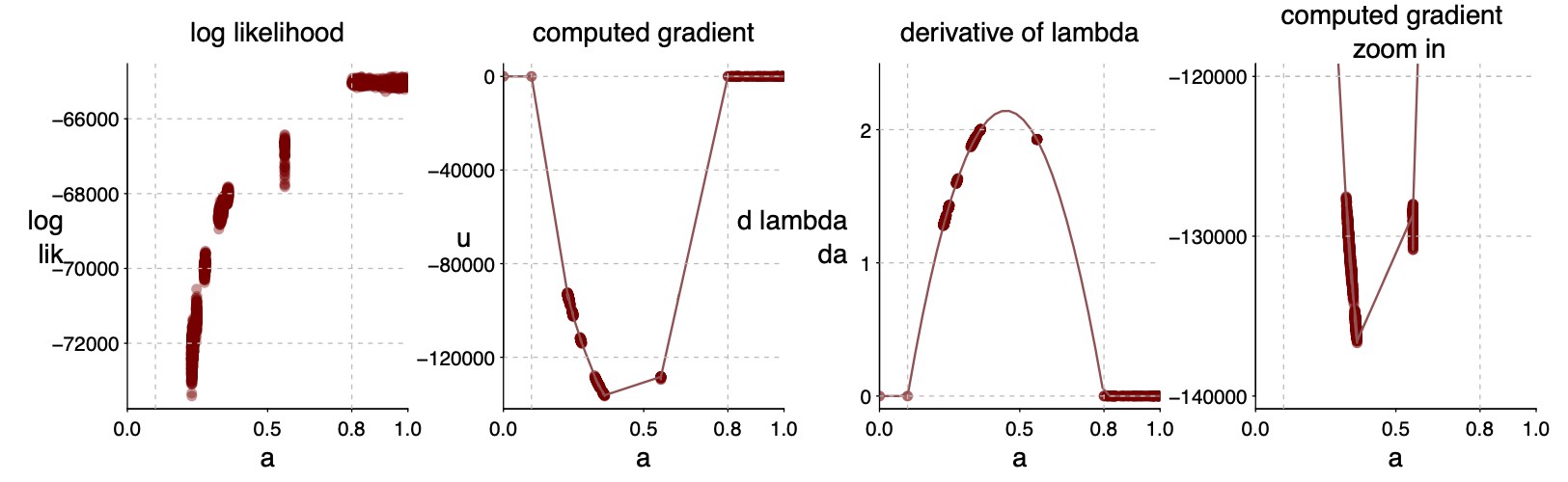}
    \caption{\em (1) the pointwise log likelihood $p(y|theta, a)$ among all simulations draws as a function of $a$.  (2) the computed pointwise gradient u. (3) the derivative ${d}\lambda$/{da}. (4) zoom in panel 2 and only present the lower end.}
    \label{fig:lda2}
\end{figure}

Measured from its relative error, the path sampling estimate of log normalizing constant has been stable after 10 adaptations (last panel in Figure \ref{fig:lda1}).  	
However, because the log normalizing constant enters the joint density in an additive way, the absolute scale of the approximation error does matter for the purpose of tempering.  Having an approximation error $\sim 10^3$ in the scale of $\log z$ is tiny compared to the range of the whole curve, and inevitable if we use any parametric regularization to fit the curve. But this  $10^3$ error is comparable with the scale of log likelihood. That is, if we start with the exactly known log normalizing constant $\log z(a) = \mathcal{O}(10^4)$, we will obtain the exact uniform margin of $a$ in the posterior. But if at one single temperature point $a_1$ we add an $1\%$ noise $ \log \hat z(a_0) += 0.01\log z(a_0) $, we will create an $\exp({100}) = 10^{43}$ bump in the marginal density $p(a)$: essentially a point mass at $a_0$.

This pitfall does not mean our method is particularly flawed. In fact, discrete tempering methods are even worse, as (a) in the best case they estimate a coarse discrete ladder, while here every point matters, and  (b) they  work in the scale of normalizing constant  $z$ directly, and it is hopeless to estimate a quantity in the scale of $\exp(-80000)$ with absolute accuracy. 

As far as we know, there has not been any attempt to run discrete tempering on LDA models. In the usual Metropolis-within-Gibbs scheme, a Gibbs swap requires draws from the conditional distribution given temperature, a step taking several hours in this large model, and discrete tempering often requires several thousand such steps!  In contrast, our method is feasible to run on this dataset because it only requires one joint HMC sample per adaptation. 

What is the general lesson we can learn from this counterexample? First, estimating the log normalizing constant is not equivalent to successful tempering, where we care more about the absolute approximation error in the second case. 

Second, tempering imposes dimensional limitations. The log likelihood scales linearly with the data input, and the magnitude depends on the how good the model fits the data. 
A generic prior-posterior geometric path will essentially fail when we add more and more data.

Third, in this example LDA is not even designed for predictions and its negative log likelihood (i.e., pointwise training error) is large even for one single input point. In general, a weak prediction model will amplify the log likelihood explosion in the prior-posterior path.  

One remedy here is to start with a better constructed base measurement, such that the log normalizing constant will be smaller.  In this example,  the discrepancy between the prior and posterior is too large, as $\mathrm{KL}$(prior, posterior) is  of the order $\exp(70000)$, and even path sampling fails to fill the gap.

For this particular model, if the goal is the log normalizing constant (or, equivalently, the log marginal likelihood), we believe the path sampling estimate  is still useful, and arguably more reliable than importance sampling and bridge sampling for reasons we have explained in our paper. But we will not use path sampling and simulated tempering for this LDA model when the goal is posterior sampling, for which we instead recommend the multi-chain-stacking  \citep{yao2020stacking}.

\end{document}